%% file: forwardEquationBarrier_Paper_Master.tex
\pdfoutput=1
\documentclass[oneside,english,british,11pt]{amsart}
\usepackage[T1]{fontenc}
\usepackage[latin9]{inputenc}
\usepackage{geometry}
\usepackage{babel}
\usepackage{verbatim}
\usepackage{float}
\usepackage{calc}
\usepackage{units}
\usepackage{mathtools}
\usepackage{amsthm}
\usepackage{amstext}
\usepackage{amssymb}
\usepackage{stmaryrd}
\usepackage{graphicx}
\usepackage{setspace}
\usepackage{esint}
\makeatletter

\providecommand{\tabularnewline}{\\}
\floatstyle{ruled}
\newfloat{algorithm}{tbp}{loa}
\providecommand{\algorithmname}{Algorithm}
\floatname{algorithm}{\protect\algorithmname}

\numberwithin{equation}{section}
\numberwithin{figure}{section}
 \theoremstyle{definition}
 \newtheorem*{defn*}{\protect\definitionname}
 \theoremstyle{assumption}
 \newtheorem*{assn*}{\protect\assumptionname}
 \theoremstyle{assumption}
 \newtheorem{assn}{\protect\assumptionname}
  \theoremstyle{remark}
  \newtheorem*{rem*}{\protect\remarkname}
\theoremstyle{plain}
\newtheorem{thm}{\protect\theoremname}
  \theoremstyle{plain}
  \newtheorem*{fact*}{\protect\factname}
  \theoremstyle{plain}
  
  \theoremstyle{plain}
  \newtheorem{prop}[thm]{\protect\propositionname}
  \theoremstyle{plain}
  \newtheorem{cor}[thm]{\protect\corollaryname}
  \usepackage{algolyx}
  \theoremstyle{plain}
  \newtheorem*{lem*}{\protect\lemmaname}

\makeatother

\makeatother

\usepackage{babel}
\usepackage[T1]{fontenc}
\usepackage{ae,aecompl,aeguill}
\usepackage{caption}
\makeatother

  \addto\captionsbritish{\renewcommand{\corollaryname}{Corollary}}
  \addto\captionsbritish{\renewcommand{\definitionname}{Definition}}
  \addto\captionsbritish{\renewcommand{\assumptionname}{Assumption}}  
  \addto\captionsbritish{\renewcommand{\factname}{Fact}}
  \addto\captionsbritish{\renewcommand{\lemmaname}{Lemma}}
  \addto\captionsbritish{\renewcommand{\propositionname}{Proposition}}
  \addto\captionsbritish{\renewcommand{\remarkname}{Remark}}
  \addto\captionsbritish{\renewcommand{\theoremname}{Theorem}}
  \addto\captionsenglish{\renewcommand{\corollaryname}{Corollary}}
  \addto\captionsenglish{\renewcommand{\definitionname}{Definition}}
  \addto\captionsenglish{\renewcommand{\factname}{Fact}}
  \addto\captionsenglish{\renewcommand{\lemmaname}{Lemma}}
  \addto\captionsenglish{\renewcommand{\propositionname}{Proposition}}
  \addto\captionsenglish{\renewcommand{\remarkname}{Remark}}
  \addto\captionsenglish{\renewcommand{\theoremname}{Theorem}}
  \providecommand{\corollaryname}{Corollary}
  \providecommand{\definitionname}{Definition}
  \providecommand{\assumptionname}{Assumption}  
  \providecommand{\factname}{Fact}
  \providecommand{\lemmaname}{Lemma}
  \providecommand{\propositionname}{Proposition}
  \providecommand{\remarkname}{Remark}
	\addto\captionsenglish{\renewcommand{\algorithmname}{Algorithm}}
	\providecommand{\theoremname}{Theorem}

\newcommand{\dx}{{\, \rm d}x}
\newcommand{\db}{{\, \rm d}b}

\newcommand{\dy}{{\, \rm d}y}

\newcommand{\du}{{\, \rm d}u}
\newcommand{\D}{{\rm d}}

\newcommand{\KuS}{K \vee S_{0}}

\let\originalleft\left
\let\originalright\right
\renewcommand{\left}{\mathopen{}\mathclose\bgroup\originalleft}
\renewcommand{\right}{\aftergroup\egroup\originalright}

\usepackage[unicode=true,
bookmarks=true,bookmarksnumbered=false,bookmarksopen=false,
breaklinks=false,pdfborder={0 0 1},backref=false,colorlinks=false]
{hyperref} 
\hypersetup{pdftitle={A Forward Equation for Barrier Options under the Brunick-Shreve Markovian Projection},
pdfauthor={Matthieu Mariapragassam}}

\begin{document}

\title[A Forward Equation for Barrier Options]{A Forward Equation for Barrier Options under the Brunick\&Shreve
Markovian Projection}

\author{Ben Hambly$^{\dagger}$, Matthieu Mariapragassam$^{\dagger}$ \& Christoph Reisinger$^{\dagger}$}

\thanks{$\!\!\!\!\! ^{\dagger}$\textsc{\small{}}
\textsc{\small{}Mathematical Institute \& Oxford-Man Institute}\\
\textsc{\small{}University of Oxford}\\
Oxford OX2 6HD|6ED, UK\\
\{ben.hambly,~matthieu.mariapragassam,~christoph.reisinger\}@maths.ox.ac.uk\textsc{\small{}$ $}\\
\\
The authors gratefully acknowledge the financial support of the \textsc{Oxford-Man
Institute of Quantitative Finance} and \textsc{BNP Paribas London}
for this research project.
}
\begin{abstract}
We derive a forward equation for arbitrage-free barrier option prices in continuous semi-martingale models, in terms of Markovian
projections of the stochastic volatility process.
This leads to a Dupire-type
formula for the coefficient derived by Brunick and Shreve for their
mimicking diffusion and can be interpreted as the canonical extension
of local volatility for barrier options. Alternatively, a forward
partial-integro differential equation (PIDE) allows computation of
up-and-out call prices under such a model,
for the complete set of strikes, barriers and maturities from a single equation.
In the same way as the vanilla forward Dupire PDE, the above-named forward
PIDE can serve as a building block for an efficient calibration routine
including barrier option quotes. We propose a discretisation scheme for the PIDE
as well as a numerical validation. 
\end{abstract}
\maketitle

\input{Chapter-1.tex}

\section*{Acknowledgements}
The authors thank Marek Musiela from the \textsc{Oxford-Man Institute},
Alan Bain and Simon McNamara from\textsc{ BNP Paribas London} as well as an anonymous referee for
their insightful comments.

\bibliographystyle{plain}
%\nocite{*}
\bibliography{BibTex/GeneralBiblio}

\appendix

\input{derivation.tex}

\end{document}

%% file: Chapter-1.tex
\section{Introduction}

Efficient pricing and hedging of exotic derivatives requires a model
which is rich enough to re-price accurately a range of liquidly traded market products.
The case of calibration to vanilla options is now widely documented and has
been considered extensively in the literature since the work of Dupire \cite{Dupire}
in the context of local volatility models (see also Gyöngy's formula \cite{Gyongy1986}).
Nowadays, the exact re-pricing of call options is a must-have standard, and LSV (local-stochastic volatility) models are the \emph{state-of-the-art} in many financial
institutions because of their superior dynamic properties over pure local volatility.
Various sophisticated calibration techniques are in use in the financial industry, for example, based on the work of Guyon and Labordère \cite{Guyon2011} as well as Ren, Madan and Qian \cite{Ren&Madan2007}. 
However, practitioners are increasingly interested in taking into
account the quotes of touch and barrier options as well; the extra
information they embed can be valuable in obtaining arbitrage-free values
of exotic products with barrier features. 

A few published works already address this question from different angles and under different assumptions.
For example, in Crosby and Carr \cite{Carr&Crosby2008} a particular
class of models %shows interesting properties that gives a method to
gives a calibration to both vanillas and barriers. 
Model-independent bounds on the price of double no-touch options were inferred from vanilla and digital option quotes in Cox and Obloj \cite{Cox2011}. Pironneau \cite{Pironneau} proves
that the Dupire equation is still valid for a given barrier level, in a local volatility setting.
The direct generalisation of this result to general stochastic volatility models appears not to be straightforward.

In our work, we approach the problem from the Brunick-Shreve
mimicking point of view \cite{BrunickShreve2013} in the general framework of continuous stochastic volatility models.
% and find a model independent solution (for continuous processes).
We derive a condition to be satisfied by
the expectation of the stochastic variance conditional on the spot
and its running maximum, $\left(S_{t},M_{t}\right)$, in order to reproduce
barrier prices. 
This conditional expectation is often referred to as a Markovian
projection.
For simplicity, we focus on up-and-out call options with no rebate and continuous monitoring of the barrier,
but the analysis extends to other payoff types. In the case of double barrier options, one would also have to consider the running minimum leading to an additional dimension and boundary term.

The derivation is inspired by the work of Derman \& Kani \cite{DermanKani1996} and Dupire \cite{Dupire, Dupire2009}.
In that context,
Gyöngy's mimicking result \cite{Gyongy1986} provides the financial engineer with a recipe
to build a low-dimensional Markovian
process which reproduces exactly any marginal density of the spot price process
for all times $t$. Brunick and Shreve \cite{BrunickShreve2013} extend Gyöngy's result to path functionals of the underlying process, of which the running maximum is a special case. In particular, they prove for any stochastic volatility process the existence of a mimicking low-dimensional Markovian process with the same joint density for $\left(S_{t},M_{t}\right)$.

In \cite{Forde2013}, Forde presents a way
to retrieve the mimicking coefficient by deriving a forward equation
for the characteristic function of $\left(S_{t},M_{t}\right)$; computation
of the coefficient is possible via an inverse two-dimensional Fourier-Laplace
transform. Here, we present an alternative method to retrieve the mimicking
coefficient, which lies closer to the well-known Dupire formula and can hence
benefit from the earlier work in the field of vanilla calibration.

We derive a partial-integro differential equation (PIDE) in strike, barrier
level and maturity for up-and-out calls priced under
the Brunick-Shreve model. This forward differential equation has the
same useful properties (and drawbacks) as the Dupire forward PDE. As a consequence,
it can be used to price a set of up-and-out calls in one single resolution.
In this regard, our method shares some similarities with the forward equations
derived by Carr and Hirsa \cite{Carr&Hirsa2007}. We highlight a few key differences though.
First, we consider the class of stochastic volatility models rather than  local volatility models
with a jump term. % with local jump compensator function.
In working directly with Brunick and Shreve's
Markovian projection onto $\left(S_{t},M_{t}\right)$,
we need to consider the running maximum explicitly in our derivation.
This complicates the proof and makes the idea used in \cite{Carr&Hirsa2007} --
employing stopping times -- not applicable in our case.
For the above type of Markovian projection, one gets an unusual ``integro'' term in the PIDE involving a second derivative, which requires particular care in the numerical solution.
In contrast, the ``integro'' term in \cite{Carr&Hirsa2007} comes from the jump process and is not treated in the same way as ours.

A very recent paper by Guyon \cite{Guyon2014}
explains how path-dependant volatility models, like the Brunick-Shreve one,
may be very useful to replicate a market's spot-volatility dynamics,
in particular highlighting the running maximum.
It is important to note, though, that the diffusion of interest can be a
fairly general stochastic process in our framework.
More specifically, the variance process does not need to contain 
the running maximum in its parametrisation, and we do not particularly advocate the dynamic use of such a model here.
Indeed, by doing so, one may find oneself with the logical conundrum that the
model for the underlying depends on the time of inception of the \emph{option}, i.e., the time the clock starts for the running maximum.
The view we take here is that the Brunick-Shreve projection is a
``code book'' (to borrow a term from \cite{carmona09}) for barrier option prices, to which other models may be calibrated. 
Additionally, the forward PIDE enables in principle the pricing over a wide set of up-and-out call deals, creating a possible efficient direct solver for the inverse problem, i.e., to retrieve model parameters from any desired model class via the projected volatility $\sigma\left(S_{t},M_{t},t\right)$. 

The remainder of this paper is organised as follows.
In Section \ref{sec:Setup-and-Main},
we introduce the modelling setup and hypotheses and derive, as our first main result, a forward equation (in terms of the maturity) for barrier option prices,
where the strike and barrier levels are spatial variables.
Next, in Section \ref{sec:A-Dupire-Type-Formula}, we derive
a Dupire-type formula for barrier options, leading to a known setup
for the reader familiar with vanilla calibration. In order to build the first step of a calibration routine to reprice up-and-out call options, we deduce a forward PIDE in Section \ref{sec:Forward-PIDE}
with better stability properties
and develop a numerical solution scheme.
Finally, in Section \ref{sec:Numerical-Results}, we show that the
forward PIDE and the ``classical'' backward pricing PDE agree on
the price of barrier options, which validates our approach.
Section \ref{sec:conclusions} concludes.

\input{mainresult.tex}

\section{A Dupire-Type Formula for Barrier Options\label{sec:A-Dupire-Type-Formula}}

The program we aim to complete is summarized in Fig.~\ref{fig:summary}.
\begin{figure}
\fbox{
\begin{tabular}{ccccc}
%\hline
&& Gy{\"o}ngy \cite{Gyongy1986} && Dupire \cite{Dupire}   \\
&& $\downarrow$ && $\downarrow$  \\
$\sigma_{|S}^{2}\left(K,t\right)$  &=& $\mathbb{E^{\mathbb{Q}}}\left[\alpha_{t}^{2}\mid S_{t}=K\right]$
&=& $\frac{\frac{\partial C}{\partial T}-\left(r-q\right)\left(C-K\frac{\partial C}{\partial K}\right)}{\frac{1}{2}K^{2}\frac{\partial^{2}C}{\partial K^{2}}}$ 
\vspace{0.4 cm} \\ 
%\hline
$\sigma_{|S,M}^{2}\left(K,B,t\right)$ &=&
$\mathbb{E^{\mathbb{Q}}}\left[\alpha_{t}^{2}\mid S_{t}=K\cap M_{t}=B\right]$
& = &
%{\bf  
(\ref{eqn:bs-mimick}) \\
&& $\uparrow$ &&   \\
&& Brunick \& Shreve \cite{BrunickShreve2013} &&
%\hline
\end{tabular}
}
\caption{Analogy of results for European and barrier options.}
\label{fig:summary}
\end{figure}
Equation (\ref{eq:MainResult}) cannot be solved for the Markovian projection $\sigma_{|S,M}$ directly, due to the presence of the 
term $\sigma_{|S,M}(B,B,T)$.
%We first rewrite (\ref{eq:MainResult}) in a different form.
We derive a formula for $\sigma_{|S,M}$ as follows.

At $K=0$, (\ref{eq:MainResult}) still holds and integration with respect to $B$ gives
\begin{equation}
\frac{\partial\widetilde{C}\left(0,B,T\right)}{\partial T}=-\frac{1}{2}\sigma_{|S,M}^{2}\left(B,B,T\right)B^{3}\left.\frac{\partial^{3}\widetilde{C}\left(K,B,T\right)}{\partial K^{2}\partial B}\right\rfloor _{K=B}\quad \forall\left(B,T\right)\in\left]S_{0},+\infty\right[\times\mathbb{R}_{*}^{+}, \label{eq:ForNT forward equation}
\end{equation}
noticing that because of (\ref{eq:Barrier-Prices-Links-With-Density}) and (\ref{zeroSzero}) no integration constant appears.
It links the price of the foreign no-touch option,
\begin{eqnarray}
\label{eqn:fnt}
FNT\left(B,T\right)=D\left(T\right)\mathbb{E^{\mathbb{Q}}}\left[S_{T}\mathbf{1}_{M_{T}<B}\right]=\frac{\widetilde{C}\left(0,B,T\right)}{Q\left(T\right)}\,,
\end{eqnarray}
to the market-implied joint-density of $\left(S_{t},M_{t}\right)$ at $K=B$; see (\ref{eq:Barrier-Prices-Links-With-Density}).
%This equation also suggests that it is probably not possible to find
%a forward equation for no-touch options alone without involving the
%joint density of spot and running maximum.
If we substitute (\ref{eq:ForNT forward equation})
into the last term of (\ref{eq:MainResult}), we get
{\small
\[
\frac{\partial^{2}\widetilde{C}\left(K,B,T\right)}{\partial B\partial T}+\left(r\left(T\right)-q\left(T\right)\right)K\frac{\partial^{2}\widetilde{C}\left(K,B,T\right)}{\partial K\partial B}=\frac{1}{2}\sigma_{|S,M}^{2}\left(K,B,T\right)K^{2}\frac{\partial^{3}\widetilde{C}\left(K,B,T\right)}{\partial K^{2}\partial B}+\frac{\partial\left(\frac{\left(B-K\right)^{+}}{B}\frac{\partial\widetilde{C}\left(0,B,T\right)}{\partial T}\right)}{\partial B}\,.
\]
}
Therefore, as $B\geq K$ in the case of interest,
{\small
\begin{eqnarray}
\frac{\partial^{2}\widetilde{C}\left(K,B,T\right)}{\partial B\partial T}-\frac{K}{B^{2}}\frac{\partial\widetilde{C}\left(0,B,T\right)}{\partial T}-\frac{\left(B-K\right)}{B}\frac{\partial^{2}\widetilde{C}\left(0,B,T\right)}{\partial T\partial B}+\left(r\left(T\right)-q\left(T\right)\right)K\frac{\partial^{2}\widetilde{C}\left(K,B,T\right)}{\partial K\partial B} & =\label{eq:Dupire-type equation reworked with DDT}\\
\frac{1}{2}\sigma_{|S,M}^{2}\left(K,B,T\right)K^{2}\frac{\partial^{3}\widetilde{C}\left(K,B,T\right)}{\partial K^{2}\partial B}\,.\nonumber 
\end{eqnarray}
}
By rearrangement of (\ref{eq:Dupire-type equation reworked with DDT}) we get the following Dupire-type formula.

\begin{cor}
Under the assumptions of Theorem \ref{thm:MainResult},
the unique Brunick-Shreve mimicking volatility for up-and-out call options is
{\small
\begin{eqnarray}
\label{eqn:bs-mimick}
\sigma_{|S,M}\left(K,B,T\right) \;=\; \sqrt{\cfrac{\frac{\partial^{2}\widetilde{C}\left(K,B,T\right)}{\partial B\partial T}-\frac{K}{B^{2}}\frac{\partial\widetilde{C}\left(0,B,T\right)}{\partial T}-\frac{\left(B-K\right)}{B}\frac{\partial^{2}\widetilde{C}\left(0,B,T\right)}{\partial T\partial B}+\left(r\left(T\right)-q\left(T\right)\right)K\frac{\partial^{2}\widetilde{C}\left(K,B,T\right)}{\partial K\partial B}}{\frac{1}{2}K^{2}\frac{\partial^{3}\widetilde{C}\left(K,B,T\right)}{\partial K^{2}\partial B}}},&& \\
 0\leq K\leq B,\, T\geq0\,.&&
 \nonumber
\end{eqnarray}
}
\end{cor}

This formula suffers from similar issues to those one finds with the evaluation of the standard Dupire formula,
in that derivatives of the value function over a continuum of strikes and maturities are required.
Practical approximations using the sparsely available data are necessarily sensitive to the method of interpolation. 

This is exacerbated here as (\ref{eq:MainResult}) involves derivatives up to order four. 
%In order to reduce the impact of this issue to some extent, 
The formula (\ref{eqn:bs-mimick}) can be expected to be numerically somewhat more stable since we
replaced the higher-order derivative with respect to the strike at barrier level with first and second order derivatives at zero strike.
In addition to the numerical improvement, it is also easier in practice to retrieve barrier prices
at zero strike with foreign no-touch prices, recalling (\ref{eqn:fnt}),
for which quotes are often available (e.g., in the FX markets). 

%As mentioned at the start of this section, the issue remains that
%one needs to interpolate between
%option prices to build a function of $\left(K,B,T\right)$ which is
%sufficiently smooth; in practice, this can lead to quite different
%Brunick-Shreve volatilities for various interpolation methods or small
%changes in input data.

\section{A Forward Partial-Integro Differential Equation For Barrier Options \label{sec:Forward-PIDE}}

It is well known that ill-posed parameter estimation
problems %of the type studied in the previous section
can often be regularised through a penalised optimisation routine,
a good example of which is the calibration of local volatility through Tikhonov regularisation as presented, e.g., by Crépey \cite{Crepey2010}, Egger and Engl \cite{Egger2005}, as well as Achdou and Pironneau \cite{Achdou2005}.

However, in order to achieve this, a suitable forward partial differential
equation is required.
%We discuss the construction and discretisation
%of such a forward equation in the following section.
We now propose a further rearrangement of (\ref{eq:MainResult}), which is more suited to the numerical computation of $C(K,B,T)$ taking the Markovian projection $\sigma_{|S,M}$ as input.\looseness=1

\subsection{Formulation as PIDE}

Equation (\ref{eq:MainResult}) can also be expressed in PIDE form
by integrating with respect to $B$. We start by integrating the diffusive
term of (\ref{eq:MainResult}),
{%\small
\[
\begin{array}{ccc}
\int_{S_{0}\lor K}^{B}\frac{1}{2}\sigma_{|S,M}^{2}\left(K,b,T\right)K^{2}\frac{\partial^{3}\widetilde{C}\left(K,b,T\right)}{\partial K^{2}\partial b}\,\db & = & \!\!\!\!\!\!\!\!\!\!\!\!\!\!\!\!\!\! \frac{1}{2}\sigma_{|S,M}^{2}\left(K,B,T\right)K^{2}\frac{\partial^{2}\widetilde{C}\left(K,B,T\right)}{\partial K^{2}}\\
&& -\frac{1}{2}\sigma_{|S,M}^{2}\left(K,S_{0}\lor K,T\right)K^{2}\frac{\partial^{2}\widetilde{C}\left(K,S_{0}\lor K,T\right)}{\partial K^{2}} \\
&& -\int_{S_{0}\lor K}^{B}\frac{1}{2}K^{2}\frac{\partial^{2}\widetilde{C}\left(K,b,T\right)}{\partial K^{2}}\frac{\partial\sigma_{|S,M}^{2}\left(K,b,T\right)}{\partial b}\,\db.
\end{array}
\]
}
The term 
in the second line
%$\frac{\partial^{2}\widetilde{C}\left(K,S_{0}\lor K,T\right)}{\partial K^{2}}$ 
vanishes for all $K$ and $T$. Indeed, when $S_{0} \geq K$ this option is already knocked-out. When $S_{0}<K$, 
%$\frac{\partial^{2}\widetilde{C}\left(K,K,T\right)}{\partial K^{2}}$ 
the term
is also zero as explained below in ($\ref{eqn:gammaKB}$).
The other terms can all be directly integrated with respect to $B$
taking into account that, similarly, no integration constant will
appear. 
This allows us to define the following initial boundary value
problem.

\begin{cor}
Under the assumptions of Theorem \ref{thm:MainResult},
the up-and-out call price %under the stochastic volatility model (\ref{eq:Model Definition})
follows the Volterra-type PIDE, 
expressed as an initial boundary value problem,
{\small
\begin{eqnarray}
\label{eq:Volettera-Type-PIDE}
\frac{\partial\widetilde{C}\left(K,B,T\right)}{\partial T}+\left(r\left(T\right)-q\left(T\right)\right)K\frac{\partial\widetilde{C}\left(K,B,T\right)}{\partial K}
-\frac{1}{2}\sigma_{|S,M}^{2}\left(K,B,T\right)K^{2}\frac{\partial^{2}\widetilde{C}\left(K,B,T\right)}{\partial K^{2}} =  \hspace{1.5 cm} &&\\
\nonumber
- \frac{1}{2}\sigma_{|S,M}^{2}\left(B,B,T\right)B^{2}\left(B-K\right)\frac{\partial^{3}\widetilde{C}\left(B,B,T\right)}{\partial K^{2}\partial B} -\int_{S_{0}\lor K}^{B}\frac{1}{2}K^{2}\frac{\partial^{2}\widetilde{C}\left(K,b,T\right)}{\partial K^{2}}\frac{\partial\sigma_{|S,M}^{2}\left(K,b,T\right)}{\partial b}\, \db\,, \hspace{-1 cm} && \\ 
\left(K,B,T\right)\in\left[0,+\infty\right[\times\left]S_{0},+\infty\right[\times\mathbb{R}_{*}^{+},\nonumber \hspace{-1 cm} &&
\end{eqnarray}
\begin{eqnarray}
\label{eqn:ic}
\widetilde{C}\left(K,B,0\right) = \left(S_{0}-K\right)^{+}\mathbf{1}_{S_{0}<B}, &\quad& T=0, \\
\label{eqn:bcK}
\widetilde{C}\left(B,B,T\right) = 0, & \quad & K=B,\\
\label{eqn:bcB}
\widetilde{C}\left(K,S_{0},T\right) = 0, & \quad & B=S_{0}.
\end{eqnarray}
}
\end{cor}
\begin{figure}[H]
\includegraphics[scale=0.23]{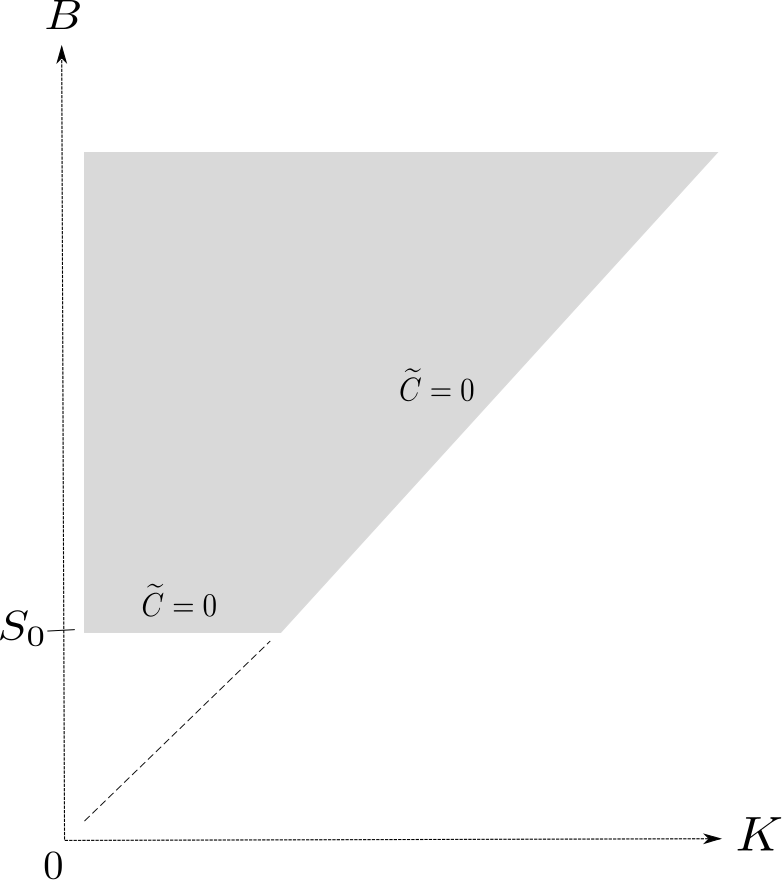}
\protect\caption{PIDE Domain and Boundaries}
\label{fig:domain}
\end{figure}

The initial condition (\ref{eqn:ic}) is the payoff obtained at maturity.
Condition (\ref{eqn:bcK}) expresses the fact that the option gets knocked out before being in-the-money if $K=B$, while (\ref{eqn:bcB}) says the option gets knocked out immediately at inception if $ S_0=B $.

\begin{rem*}
We note that in the models we will consider, no boundary condition needs to be specified at $K=0$, since the coefficients of the $K$-derivatives vanish sufficiently fast as $K\rightarrow 0$. This is the case, e.g., if $\sigma_{|S,M}^2$ and %$\frac{\partial\sigma^{2}_{|S,M}}{\partial B}$ 
its $B$-derivative
are bounded when $K \rightarrow 0$.

In fact, a boundary condition at $K=0$ could be derived by differentiating (\ref{eqn:intcall}) twice,
\begin{eqnarray*}
\frac{\partial^{2}\widetilde{C}\left(0,B,T\right)}{\partial K^{2}} & = & D\left(T\right)Q\left(T\right)\int_{ S_{0}}^{B}\phi\left(0,y,T\right) \dy =D\left(T\right)Q\left(T\right)\mathbb{Q}\left(S_0<M_{T}<B\mid S_{T}=0\right)
\psi(0,T),
\end{eqnarray*}
where 
$\psi(0,T)$
is the density of $S_{T}$ at 0. The latter is equal to zero since the
density of the spot at zero is zero under the log-spot model in (\ref{eq:Model Definition}).
Linearity conditions, also known as ``Zero
Gamma'' for spot PDEs, are commonly used in finance and particularly useful for exotic contracts. They usually lead to stable numerical schemes (see Windcliff \cite{WIndcliff2004}).
%\begin{rem*}

We also note that
\begin{eqnarray}
\label{eqn:gammaKB}
\frac{\partial^{2}\widetilde{C}\left(B,B,T\right)}{\partial K^{2}}=D\left(T\right)Q\left(T\right)\mathbb{Q}\left(S_0<M_{T}<B\mid S_{T}=B\right) 
\psi(B,T) = 0,
\end{eqnarray}
because if the spot is equal to $B$ at $T$, then the probability
that the running maximum is smaller than $B$ is zero. While we
will use the Dirichlet condition on the boundary $K=B$ for computations,
%\begin{flushleft}
(\ref{eqn:gammaKB}) allows us to express the third order cross derivative in the second line of (\ref{eq:Volettera-Type-PIDE}) by a third order derivative only with respect to $K$.
Indeed,
\[
0 = \frac{\D}{\D B}\left(\frac{\partial^{2}\widetilde{C}\left(B,B,T\right)}{\partial K^{2}}\right)=\frac{\partial^{3}\widetilde{C}\left(B,B,T\right)}{\partial K^{3}}+\frac{\partial^{3}\widetilde{C}\left(B,B,T\right)}{\partial K^{2}\partial B},
\]
leading to
%\par\end{flushleft}
\begin{eqnarray}\label{eqn:thirdder}
\frac{\partial^{3}\widetilde{C}\left(B,B,T\right)}{\partial K^{3}}=-\frac{\partial^{3}\widetilde{C}\left(B,B,T\right)}{\partial K^{2}\partial B}.
\end{eqnarray}
\end{rem*}

The domain sketched in Fig.~\ref{fig:domain} is unbounded for large $B$.
We denote $B_{Max}$ the maximum value of barrier levels we will consider in the numerical solution. Note that because of the structure of (\ref{eq:Volettera-Type-PIDE}),
no boundary condition is needed for $B=B_{Max}$.

\subsection{Finite Difference Approximation}
\label{sec:fwpricing}

We define a uniform mesh which contains $M+1$ time points, $N+1$
spatial points in the strike and $P+1$ in the barrier variable, leading to the
following definition of the step sizes:
\begin{eqnarray*}
\begin{array}{lclll}
K_{i} &=& i\Delta_{K}, &\Delta_{K}=\frac{B_{Max}-S_{0}}{N+1},& i\in\left\llbracket 0,N\right\rrbracket, \\
B_{j} &=& S_{0}+j\Delta_{B}, &\Delta_{B}=\frac{B_{Max}-S_{0}}{P+1},& j\in\left\llbracket 0,P\right\rrbracket, \\
T_{m} &=& m\Delta_{T}, &\Delta_{T}=\frac{T_{Max}}{M+1},& m\in\left\llbracket 0,M\right\rrbracket.
\end{array}
\end{eqnarray*}
For simplicity, we impose $\Delta_{K}=\Delta_{B}$.
This will ensure that for any $B_{j}$, the corresponding mesh line
will contain at least all $(B_{u},B_j)$ for all $u$ smaller than $j$, which is useful for the following algorithm.

We can identify an interesting property of (\ref{eq:Volettera-Type-PIDE}),
especially visible with the substitution (\ref{eqn:thirdder}),
in order to approximate the solution inductively on a discrete lattice.
At $B=S_{0}$, the solution
is known to be uniformly zero over the strike and time variable. Assume now an approximate solution is known up to a certain $B$. Moving from
$B$ to $B+\Delta_{B}$, and approximating the integral by a quadrature
rule, gives a PDE in time and strike at level $B+\Delta_{B}$, which
can be solved by finite differences. From now on, we will refer to the PDE
at a given barrier level $B_{j}$ as a ``PDE layer''. We can then solve the
PIDE for each layer at points $B_{j}$ from $S_{0}$ to $B_{Max}$.

We denote the discrete solution vector in such a layer by 
\[
\mathbf{u}_{.,j}^{m}=\begin{bmatrix}\widetilde{C}\left(K_{0},B_{j},T_{m}\right), \ldots,
%\vdots\\
%\vdots\\
%\vdots\\
\widetilde{C}\left(B_{j},B_{j},T_{m}\right)
\end{bmatrix}'
\]
of size $n_{j}=\frac{B_{j}-K_{0}}{\Delta_{K}}+1$, where $'$ denotes the transpose.
We also denote by $\mathbf{I}_{n}$ the identity matrix of size $n\times n$.

\subsubsection*{Derivatives}

Derivatives are approximated by centered finite differences at each space
point except at $K=0$ and $K=B_{j}$, where they are computed, respectively,
forward and backward. For the time being, the boundary conditions
are not taken into account. We assume equally spaced points and define
two operators as follows:
\begin{eqnarray}
\delta_{K}\mathbf{u}_{i,j}^{m} &=& \frac{\mathbf{u}_{i+1,j}^{m}-\mathbf{u}_{i-1,j}^{m}}{2\Delta_{K}}, \nonumber \\
\delta_{KK}\mathbf{u}_{i,j}^{m} &=&\frac{\mathbf{u}_{i+1,j}^{m}-2\mathbf{u}_{i,j}^{m}+\mathbf{u}_{i-1,j}^{m}}{\Delta_{K}^{2}}. \label{eq:PIDE_Derivatives_FD}
\end{eqnarray}
{} The usual matrix derivative operator can be defined for both the
first and second order derivative:
\[
\mathbf{D}=\frac{1}{\Delta_{K}}\begin{bmatrix}-1 & 1 & 0 & ... & 0\\
-\frac{1}{2} & 0 & \frac{1}{2} & \ddots & \vdots\\
0 & \ddots & \ddots & \ddots & 0\\
\vdots & \ddots & -\frac{1}{2} & 0 & \frac{1}{2}\\
0 & ... & 0 & -1 & 1
\end{bmatrix}, \quad
\mathbf{D_2} = \frac{1}{\Delta_{K}^{2}}\begin{bmatrix}1 & -2 & 1 & 0 & ... & 0\\
1 & -2 & 1 & 0 & \left(0\right) & \vdots\\
0 & \ddots & \ddots & \ddots & \ddots & \vdots\\
\vdots & \ddots & \ddots & \ddots & \ddots & 0\\
\vdots & \left(0\right) & 0 & 1 & -2 & 1\\
0 & ... & 0 & 1 & -2 & 1
\end{bmatrix}.
\]
We also define the forward time difference operator
\[
\delta_{T}\mathbf{u}_{i,j}^{m}=\frac{\mathbf{u}_{i,j}^{m+1}-\mathbf{u}_{i,j}^{m}}{\Delta T}.
\]

\subsubsection*{Integral}
The integral term will be computed using the trapezoidal quadrature
rule. This yields an order two consistent approximation. The term
of interest is
\begin{eqnarray*}
F\left(K_{i},B_{j},T_{m}\right) &=& \int_{S_{0}\lor K}^{B_{j}}-\frac{1}{2}K_{i}^{2}\frac{\partial^{2}\widetilde{C}\left(K_{i},b,T_{m}\right)}{\partial K^{2}}\frac{\partial\sigma_{|S,M}^{2}\left(K_{i},b,T_{m}\right)}{\partial b} \db.
\end{eqnarray*}
Let us assume that we know an approximation to the solution of the PIDE for the discrete
set of barriers $\left(B_{n}\right)_{n<j}$. We also know that, when
the barrier is at the spot level, the integrand is zero.
Hence,
\begin{eqnarray}
\nonumber
F\left(K_{i},B_{j},T_{m}\right) & = & - \,\, \sum_{n=1}^{j-1}\left(\frac{1}{2}K_{i}^{2}
\frac{\partial^{2}\widetilde{C}\left(K_{i},B_n,T_{m}\right)}{\partial K^{2}}
%\delta_{KK}\mathbf{u}_{i,n}^{m}
\frac{\partial\sigma_{|S,M}^{2}\left(K_{i},B_{n},T_{m}\right)}{\partial B}\Delta_{B}\right)\\
\label{eqn:corrder}
 &  & -  \,\, \frac{1}{4}K_{i}^{2}
 \frac{\partial^{2}\widetilde{C}\left(K_{i},B_j,T_{m}\right)}{\partial K^{2}}
 \frac{\partial\sigma_{|S,M}^{2}\left(K_{i},B_{j},T_{m}\right)}{\partial B}\Delta_{B}\\
 &  & +  \,\, \mathcal{O}\left(\Delta_{B}^{2}\right).
 \nonumber
\end{eqnarray}
The sum in the first line can be computed for all $K_{i}$ and $T_{m}$
in a forward induction over $j$,
as the solution is known for all the barrier levels involved. This
sum is then treated as a source function for the PDE layer of level
$B_{j}$. We define a vector 
\[
\mathbf{f}_{.,j}^{m}=\left[\begin{array}{c}
-\sum_{n=1}^{j-1}\left(\frac{1}{2}K_{i}^{2}\delta_{KK}\mathbf{u}_{i,n}^{m}\frac{\partial\sigma_{|S,M}^{2}\left(K_{i},B_{n},T_{m}\right)}{\partial B}\Delta_{B}\right)\end{array}\right]_{i=0,1,...,n_{j}}.
\]
The remaining term in (\ref{eqn:corrder}) gives a small correction to the diffusion at $B_{j}$ and we can incorporate it in the discretisation of the corresponding diffusive term
of (\ref{eq:Volettera-Type-PIDE}).

\subsubsection*{Boundary Derivative Term}

To approximate the ``boundary derivative'' at $(B,B,T)$ in (\ref{eq:Volettera-Type-PIDE}),
we use (\ref{eqn:thirdder}) and a first order approximation to the third derivative
with discretisation matrix written as
\[
\mathbf{\Phi}=\frac{1}{\Delta_{K}^{3}}\begin{bmatrix}0 & \ldots & 0 & \text{0} & -1 & 3 & -3 & 1\\
\vdots & \left(0\right) & \vdots & \vdots & \vdots & \vdots & \vdots & \vdots\\
0 & \ldots & 0 & 0 & -1 & 3 & -3 & 1
\end{bmatrix}.
\]
As this term is present in the discretised equation for all interior mesh points, this reduces the consistency order in $\Delta_{K}$ of the overall scheme to one.
We found higher order finite differences to be unstable.
\begin{rem*}
\emph{A second order accurate stable scheme can be obtained by using (\ref{eq:Barrier-Prices-Links-With-Density}) to replace the third order derivative at the boundary by the density, which can be found by solving the Kolmogorov forward equation (\ref{eqn:kfe}) numerically.}
\end{rem*}

%and the accurate and stable approximation of this term is a topic for further research.

\subsubsection*{PIDE in Terms of Matrix Operations}

If we take into account the finite difference approximations and quadrature rule for the integral,
it is now possible to give a discretised PIDE,
for a given triplet $\left(i,j,m\right)\in\left\llbracket 0,N\right\rrbracket \times\left\llbracket 0,M\right\rrbracket \times\left\llbracket 0,P\right\rrbracket $ by
\begin{eqnarray}
\delta_{T}\mathbf{u}_{i,j}^{m}+\left(r\left(T_{m}\right)-q\left(T_{m}\right)\right)K_{i}\delta_{K}\mathbf{u}_{i,j}^{m}-\frac{1}{2}\left(\sigma_{|S,M}^{2}\left(K_{i},B_{j},T_{m}\right)-\frac{1}{2}\frac{\partial\sigma_{|S,M}^{2}\left(K_{i},B_{j},T_{m}\right)}{\partial B}\Delta_{B}\right)K_{i}^{2}\delta_{KK}\mathbf{u}_{i,j}^{m}\nonumber \\
+\frac{1}{2}\sigma_{|S,M}^{2}\left(B_{j},B_{j},T_{m}\right)B_{j}^{2}\left(B_{j}-K\right)^{+}\delta_{KKB}^{-}\mathbf{u}_{n_{j},j}^{m} & =\label{eq:Discretised PIDE}\\
-\sum_{n=1}^{j-1}\frac{1}{2}K_{i}^{2}\delta_{KK}\mathbf{u}_{i,j}^{m}\frac{\partial\sigma_{|S,M}^{2}\left(K_{i},B_{n},T_{m}\right)}{\partial B}\Delta_{B},\nonumber 
\end{eqnarray}
and specify the coefficient matrices
\begin{eqnarray*}
\mathbf{A}_{.,j}^{m} & = & \left(r\left(T_{m}\right)-q\left(T_{m}\right)\right)\text{diag}\left(K_{0},...,K_{n_{j}}\right), \\
\mathbf{B}_{.,j}^{m} & = & -\frac{1}{2}\text{diag}\left(\left(\sigma_{|S,M}^{2}\left(K_{i},B_{j},T_{m}\right)K_{i}^{2}-\frac{1}{2}K_{i}^{2}\frac{\partial\sigma_{|S,M}^{2}\left(K_{i},B_{j},T_{m}\right)}{\partial B}\Delta_{B}\right)_{i\in\left\llbracket 0,n_{j}\right\rrbracket }\right), \\
\mathbf{C}_{.,j}^{m} & = & -\frac{1}{2}\text{diag}\left(\left(\sigma_{|S,M}^{2}\left(B_{j},B_{j},T_{m}\right)B_{j}^{2}\left(B_{j}-K_{i}\right)^{+}\right)_{i\in\left\llbracket 0,n_{j}\right\rrbracket }\right).
\end{eqnarray*}

\begin{rem*}
We can also approximate $\mathbf{B}_{.,j}^{m}$ further by Taylor expansion,
\[
\mathbf{B}_{.,j}^{m}=-\frac{1}{2}\text{diag}\left(\left(\sigma_{|S,M}^{2}\left(K_{i},B_{j}-\frac{\Delta_{B}}{2},T_{m}\right)K_{i}^{2}\right)_{i\in\left\llbracket 0,n_{j}\right\rrbracket }\right)+\mathcal{O}\left(\Delta_{B}^{2}\right),
\]
which has a negative sign irrespective of the mesh size and does not alter the convergence order.

\end{rem*}

Under forward Euler time stepping, the complete scheme can be more compactly written as
\begin{eqnarray*}
%\begin{cases}
\frac{\mathbf{u}_{.,j}^{m+1}-\mathbf{u}_{.,j}^{m}}{\Delta T}+\mathbf{L}_{.,j}^{m}\mathbf{u}_{.,j}^{m}&=&\mathbf{f}_{.,j}^{m}, \\
\mathbf{\mathbf{L}}_{.,j}^{m}&=&\mathbf{A}_{.,j}^{m}\mathbf{D}+\mathbf{B}_{.,j}^{m}\mathbf{D_{2}}+\mathbf{C}_{.,j}^{m}\mathbf{\Phi}.
%\end{cases}%\label{eq:PIDE-as-system-of-ODEs}
\end{eqnarray*}
Under $\theta$-time-stepping, the scheme becomes
\[
\left(\mathbf{I}_{n_{j}}+\theta\Delta_{T}\mathbf{L}_{.,j}^{m+1}\right)\mathbf{u}_{.,j}^{m+1}=\left(\mathbf{I}_{n_{j}}-\left(1-\theta\right)\Delta_{T}\mathbf{L}_{.,j}^{m}\right)\mathbf{u}_{.,j}^{m}+\theta\Delta_{T}\mathbf{f}_{.,j}^{m+1}+\left(1-\theta\right)\Delta_{T}\mathbf{f}_{.,j}^{m}.
\]
This includes the second-order accurate Crank-Nicolson scheme for $\theta=0.5$, which is used for our numerical computations.

This section shows that we can adapt the classical tools of finite difference
methods in order to solve this PIDE. Each of the layers being a one-dimensional
PDE, it is solved by successive roll-forward performed by a Gaussian elimination.
Although the Thomas algorithm cannot be applied directly
since $\mathbf{L}_{.,j}^{m}$ is not tri-diagonal (even after applying
the necessary boundary conditions), a sparse Gaussian elimination
will still be $\mathcal{O}\left(n\right)$ as the dense sub-matrix
has a fixed number of columns.

\subsubsection*{Solution Algorithm}

We conclude with a summary of a possible algorithm to solve the PIDE numerically:

\begin{algorithm}[H]
\protect\caption{PIDE Discretisation}

\begin{raggedright}
\emph{$\mathbf{u}_{.,j}^{0}$ = $\left(\left(S_{0}-K_{i}\right)^{+}\mathbf{1}_{S_{0}<B_{j}}\right)_{i\in\left\llbracket 0,N\right\rrbracket }$}
\par\end{raggedright}

\begin{raggedright}
$\mathbf{u}_{.,0}^{m}=0$
\par\end{raggedright}

\begin{raggedright}
$\mathbf{f}_{.,0}^{m}=0$
\par\end{raggedright}
\begin{algor}
\item [{for}] ( $j=0\,;\, j\leq P\,;\, j++$) 

\begin{algor}
\item [{{*}}] \textbf{solve} $B_{j}-$layer PDE for \emph{$\left(\mathbf{u}_{.,j}^{m}\right)_{m\in\left\llbracket 0,M\right\rrbracket }$:}
\[
\left(\mathbf{I}_{n_{j}}+\theta\Delta_{T}\mathbf{L}_{.,j}^{m+1}\right)\mathbf{u}_{.,j}^{m+1}=\left(\mathbf{I}_{n_{j}}-\left(1-\theta\right)\Delta_{T}\mathbf{L}_{.,j}^{m}\right)\mathbf{u}_{.,j}^{m}+\theta\Delta_{T}\mathbf{f}_{.,j}^{m+1}+\left(1-\theta\right)\Delta_{T}\mathbf{f}_{.,j}^{m}
\]

\item [{{*}}] \textbf{compute} $\mathbf{f}_{.,j+1}^{m}$ from $\mathbf{f}_{.,j}^{m}$
and $\mathbf{u}_{.,j}^{m}$ (for $j<P$)
\end{algor}
\item [{endfor}]~\end{algor}
\end{algorithm}

\medskip

\section{Numerical Results\label{sec:Numerical-Results}}

\subsection{Pricing Under the Mimicking Brunick-Shreve Model}
\label{sec:bwpricing}

For validation purposes of the forward equation, we will use a numerical solution of the backward
pricing PDE in the Brunick-Shreve model.
The backward PDE gives the price over a range of $S$ and $M$ and $t$, for fixed $K$, $B$ and $T$.
%or if one desires to use a calibrated Brunick-Shreve
%model for pricing, it is important to be able to price a barrier contract
%under the above-named process. In this section, 
%\medskip

\subsubsection*{Augmented State Feynman-Kac PDE\label{sec:Feynman-Kac-Equation-of-Barrier-Option}}

Recall that the spot diffusion
of the underlying under $\mathbb{Q}$ is
\begin{eqnarray}
\label{eqn:bwbs}
%\begin{array}{c}
\cfrac{dS_{t}}{S_{t}} = \left(r\left(t\right)-q\left(t\right)\right)\, dt+\sigma\left(S_{t},M_{t},t\right) \, dW_{t}. \
%end{array}
\end{eqnarray}

We note that the spot process is not Markovian anymore due to the dependence of the volatility on the running maximum. Hence, we cannot use the standard one dimensional Feynman-Kac PDE with Dirichlet boundary condition, but need to augment the state space even for barrier options (as well as European options). Assume that strike \textbf{$K$}, barrier $B$ and maturity $T$ are
all fixed. Then, following identical steps to the derivation by Shreve \cite{Shreve2008} in the Black-Scholes case,
we get the following (see Appendix \ref{sec:Derivation-of-the-Kolmogorov-Equations}).

\begin{prop}
\label{prop:bw}
Under a Brunick-Shreve model (\ref{eqn:bwbs}),
the up-and-out call price $C$
defined in the region $\small{\Omega_{B} = \{(x,y): 0<x<y, S_0<y<B\}}$ is the solution to the following initial boundary value problem:
\begin{eqnarray}
\label{eq:Brunick_Backward_PDE}
\frac{\partial C}{\partial t}+\left(r\left(t\right)-q\left(t\right)\right)x\frac{\partial C}{\partial x}+\frac{1}{2}x^{2}\sigma^{2}\left(x,y,t\right)\frac{\partial^{2}C}{\partial x^{2}}-r\left(t\right)C=0,\, && (x,y) \in \Omega_{B}, \;0<t<T, \\
C\left(0,y,t\right)=0, && t\leq T,\nonumber \\
C\left(x,y,T\right)=\left(x-K\right)^{+}\mathbf{1}_{y<B}, && y>S_0, \nonumber \\
\left.\frac{\partial C\left(x,y,t\right)}{\partial y}\right\rfloor _{x=y}=0, && y>S_0, \nonumber \\
C\left(x,B,t\right)=0, && x>0,\, t\leq T. \nonumber
\end{eqnarray}

\end{prop}
\subsubsection*{Finite Difference Approximation}

We briefly explain the numerical solution of (\ref{eq:Brunick_Backward_PDE}).

% by a Neumann boundary condition on a
We first notice that (\ref{eq:Brunick_Backward_PDE})
is a classical Black-Scholes PDE (where the volatility is a function of $S$, a ``parameter''
$M$, and $t$), and
%We notice that the PDE is not ``genuinely'' two-dimensional as 
derivatives with respect to $y$ only enter via
the Neumann boundary condition (on the diagonal $x=y$).
This means that
for a fixed $y$, we have to solve a Black-Scholes type PDE on
$\left[0,y\right]\ \times \{y\}$. From now on, the one-dimensional
PDE for a given level of $y$ will be referred to as a ``PDE layer''.

If $y=B$, then $C\left(x,y,t\right)=0$ for all $x$. We can then use the Neumann
boundary condition to bootstrap backwards in $y$ from $B$ to $S_{0}$.
To illustrate the idea, % by a backward finite difference,
consider now $y=B-\Delta_y$ for a mesh size $\Delta_y$,
% (although we will use higher order interpolation later),
then
\[
\left.\frac{\partial C\left(x,y,t\right)}{\partial y}\right\rfloor _{x=y}=\frac{C\left(y,y+\Delta_{y},t\right)-C\left(y,y,t\right)}{\Delta y}+\mathcal{O}\left(\Delta_{y}\right)=0\implies C\left(y,y,t\right)=C\left(y,y+\Delta_{y},t\right)+\mathcal{O}\left(\Delta_{y}^{2}\right).
\]
%Since $C\left(B,B,t\right)=0$, 
We can build the next PDE layer at $y$ using an approximate Dirichlet boundary
condition at $x=y$, $C\left(y,y,t\right)\approx C\left(y,y+\Delta_{y},t\right)$ (here, 0).
Subsequent PDE layers can be
constructed similarly by backward reasoning. We will describe a higher-order version below.
The solution in each layer will then depend on layers with greater $y$ via this boundary condition.
The premium value is retrieved from $C\left(S_{0},S_{0},0\right)$.

We now describe the construction of meshes of the form $(x_j(y),y)_{j=1,\ldots,N_x(y)}$ for layer $y$, as well as of the spacing of points in the $y$-direction.
The best accuracy was achieved numerically with a refined running maximum grid close to $S_{0}$ and, for each $y$, a uniformly spaced spot grid
(except for nodes close to the diagonal, as detailed below).

Denote by $\Delta_{y}$ a desired target mesh width. Then for the $y$-mesh we choose
$N_{y}=\left\lceil \frac{B-S_{0}}{\Delta_{y}}\right\rceil $ points
with an exponential grading as defined in \cite{White2013}, so that
$\mbox{\ensuremath{\forall\, i}\ensuremath{\in\left\llbracket 0,N_{y}\right\rrbracket }}$:
\begin{eqnarray*}
y_{i} & = & \left(S_{0}-\theta\right)+\theta\,\exp\left(\lambda z_{i}\right),\\
\lambda & = & 2,\\
\theta & = & \frac{B-S_{0}}{e^{\lambda}-1},\\
z_{i} & = & \frac{i}{\left(N_{y}+1\right)}.
\end{eqnarray*}

%\medskip{}

The construction of the mesh, illustrated in Figure \ref{fig:bwmesh},
ensures that $y_{i}$ lies on the grid of the PDE layer
of levels $y_{i+1}$ and $y_{i+2}$,
and hence no interpolation is needed to retrieve the boundary condition for the next PDE layer
as described
in the next paragraph.
%Apart from the last two points in each layer, the mesh points are spaced uniformly.

%\begin{minipage}[t]{1\columnwidth}%
\begin{figure}[h]
\includegraphics[scale=0.2]{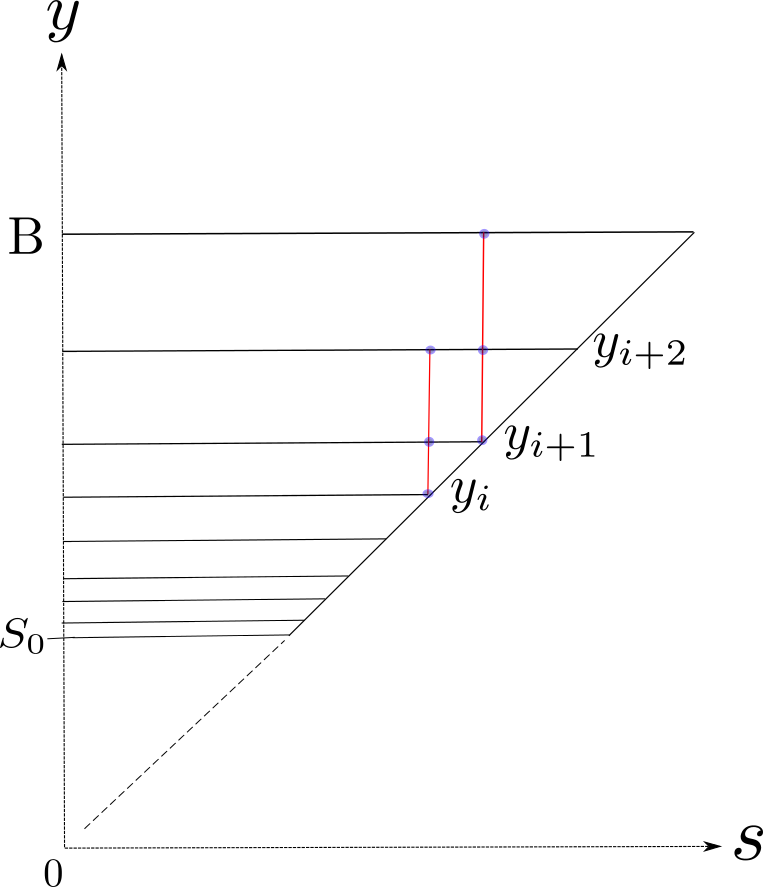}
%\protect
\caption{Mesh construction for backward equation. The spot grid is recomputed for every PDE
layer to be $(x_j(y_i),y_i)$, where $x_j = j \Delta_x$, $j<N_x-2$,
$x_{N_x-j}=y_{i-j}$, $j\in \{0,1,2\}$, $\Delta_x = y_{i-2}/(N_x-2)$ and
%$N_{x}=\left\lceil \frac{y_{i}}{\Delta_{y}}\right\rceil $
$N_{x}=\left\lceil y_{i}/\Delta_{y}\right\rceil $ 
to keep the spacing approximately constant.}
\label{fig:bwmesh}
\end{figure}
%\end{minipage}

%\begin{rem*}
In order to increase accuracy from the numerical boundary condition described at the start of this section, 
%the boundary condition in the numerical tests was computed 
we use second order Taylor expansion in $y$-direction instead of
a simple backward finite difference for all $y_i$ with $i<N_y-1$. The bootstrap idea stays
the same. Indeed, write $c\left(y\right)=c\left(y_{i},y,t\right)$
where $y_{i}$ is a certain given level of the $y$-discretisation
and $c'\left(y\right)=\frac{\partial c\left(y_{i},y,t\right)}{\partial y}$.
We consider a Taylor expansion around $y=y_{i}$: $c\left(y\right)=c\left(y_{i}\right)+\frac{1}{2}\left(y-y_{i}\right)^{2}c''\left(y_{i}\right)+o\left((y-y_i)^{2}\right)$
(as $c'\left(y_{i}\right)=0$ by the boundary condition). Using this
at points $y_{i+1}$ and $y_{i+2}$, we get $c\left(y_{i}\right)=c\left(y_{i+2}\right)+\frac{\left(y_{i+2}-y_{i}\right)^{2}}{\left(y_{i+2}-y_{i}\right)^{2}-\left(y_{i+1}-y_{i}\right)^{2}}\left(c\left(y_{i+1}\right)-c\left(y_{i+2}\right)\right)$. 
This is then used as a boundary value at $y_i$. The first layer at $N_y-1$ is treated with a first order backward difference as described previously.
%\end{rem*}

\subsection{Numerical Validation}

Our numerical validation consists in pricing a set of up-and-out call
options for different strikes, maturities and barriers with:
\begin{enumerate}
\item the Forward PIDE (\ref{eq:Volettera-Type-PIDE}) and one numerical
solution for the whole set of deal parameters;
\item the Backward Feynman-Kac PDE (\ref{eq:Brunick_Backward_PDE}) and
as many solutions as triplets of deal parameters.
\end{enumerate}
The goal is to make them match with about one basis point tolerance.

The Brunick-Shreve volatility is generated with an SVI parametrisation
in both barrier and strike dimension (details can be found in \cite{Gatheral2006})
defined as

%{\small
\begin{eqnarray}
%\hspace{-0.7 cm}
%\begin{cases}
\label{eqn:svi}
\sigma\left(x,y,t\right) &=& \frac{1}{2}\left(\sigma_{SVI}\left(\log\left(\frac{x}{S_{0}}\right),t+1\right)+\sigma_{SVI}\left(\log\left(\frac{y}{S_{0}}\right),t+1\right)\right), \\
\sigma_{SVI}\left(k,t\right)&=&\sqrt{a+b\left(\rho\left(k-m\right)+\sqrt{\left(k-m\right)^{2}+\sigma^{2}}\right)}.
\nonumber
%\end{cases}
\end{eqnarray}
%}
The Brunick-Shreve volatility surface has a shape as in Figure \ref{Brunick VolSurf Figure}.

%\begin{minipage}[t]{1\columnwidth}
\begin{figure}[H]
\includegraphics[scale=0.33]{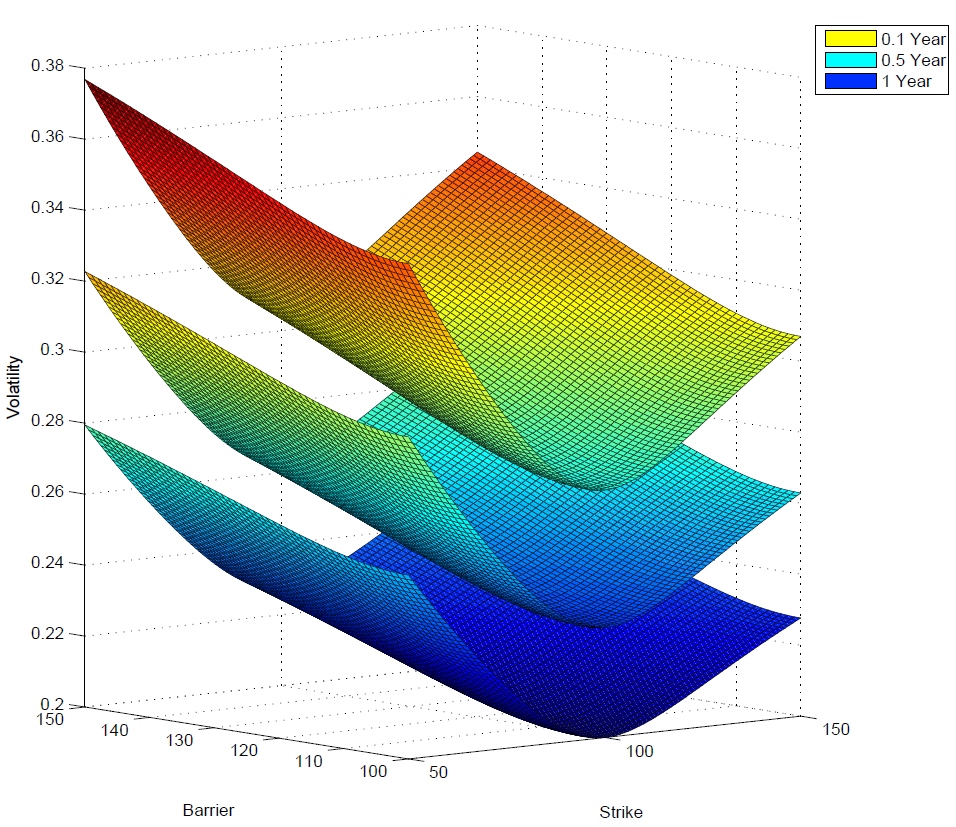}
\protect\caption{Assumed Brunick-Shreve volatility surface as per (\ref{eqn:svi}), where $a=0.04,b=0.2,\sigma=0.2,\rho=m=0$,
and $T\in \{0.1, 0.5, 1\}$.}
\label{Brunick VolSurf Figure}
\end{figure}
%
%\end{minipage}

The spot is $S_0 = 100$, the risk-free rate is $r=0.1$ and the dividend
yield is $q=0.05$.

We use the numerical schemes as described in 
\ref{sec:fwpricing}
and
\ref{sec:bwpricing}
for the PIDE and PDE solution, respectively, with Crank-Nicolson time stepping.
The discretisation parameters are
$N=P=N_y=1000$ space steps (with $N_x$ adjusted as described) and $M=1000$ time steps.

We compare prices for $\left(K,B,T\right)$ covering the set $\left[0,120\right]\times\left[100,120\right]\times\left\{ 1\right\} $
with 120 points in strike and 40 points in barrier levels, see Table \ref{table:errs}. The error
is computed as the relative error if the net present value (NPV) is
above one and as absolute error otherwise.

\begin{center}
\begin{tabular}{|c|c|}
\hline 
{\footnotesize{}Average Difference} & {\footnotesize{}Maximum Difference}\tabularnewline
\hline 
\hline 
\selectlanguage{english}%
4.6e-5\selectlanguage{british}%
 & \selectlanguage{english}%
3.5e-4\selectlanguage{british}%
\tabularnewline
\hline 
\end{tabular}\label{tab:general-figures-errors}\captionof{table}{Difference between forward and backward solutions over strikes between 0 and 120, and up-and-out barriers between 100 and 120, all for maturity 1.}
\label{table:errs}
\par\end{center}

\vspace{1em}

We can analyse more precisely the behaviour for a few barrier levels.
For example, Figure \ref{PIDEvsPDE} shows the difference as a
function of strike. The associated values for a barrier fixed at 120
are in Table \ref{tab:Errors-on-a-strike-ladder}.

%We used 1000 space steps and 1000 time steps for both the Forward PIDE and the Backward PDE.

\begin{minipage}[t]{1\columnwidth}%
\begin{minipage}[t]{0.49\columnwidth}%
\begin{figure}[H]
\vspace*{0bp}

\begin{centering}
\includegraphics[scale=0.31]{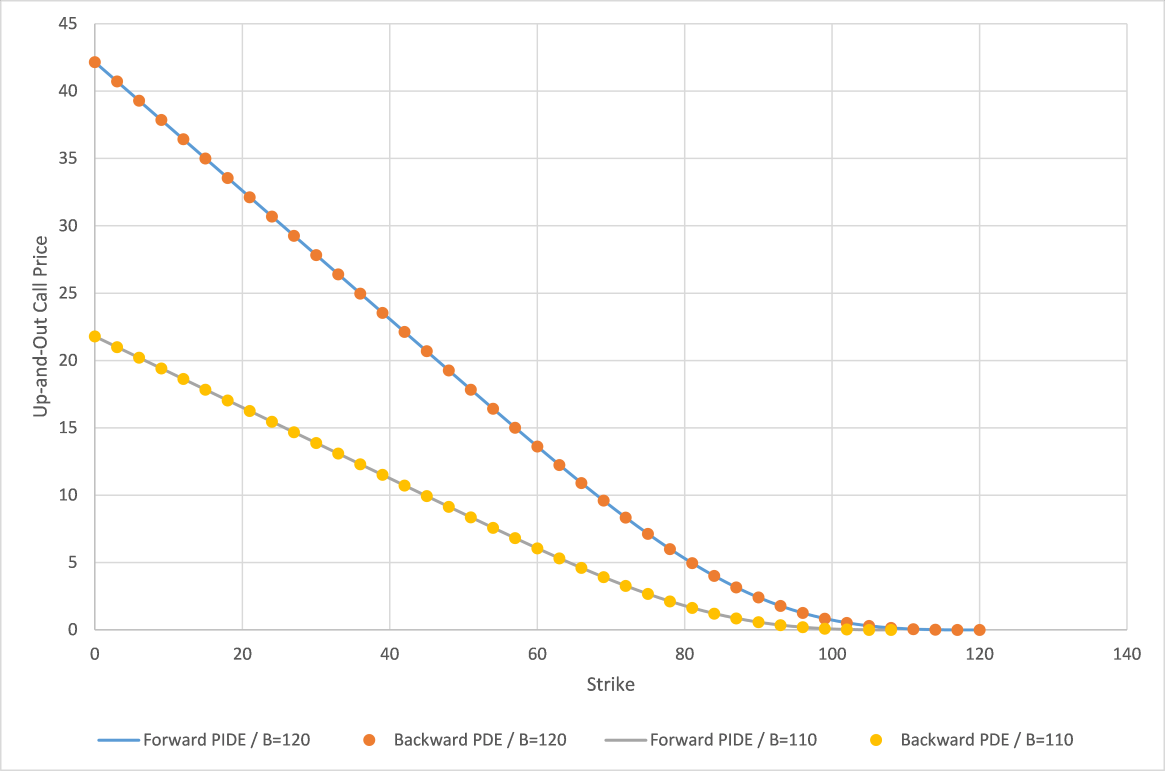}
\par\end{centering}

\protect\caption{Forward PIDE vs Backward PDE:\protect \\
up-and-out call option prices with Brunick-Shreve volatility, $S=100$
and $T=1$.\protect}

\label{PIDEvsPDE}
\end{figure}
\end{minipage}%
\begin{minipage}[t]{0.49\columnwidth}%
\vspace*{0bp}

\begin{center}
\begin{tabular}{|c|c|c|c|}
\hline 
Strike & Forward & Backward  & Rel. \tabularnewline
 & PIDE & PDE & Diff.\tabularnewline
\hline 
\hline 
0 & 42.1486 & 42.1486 & 1e-06\tabularnewline
\hline 
9 & 37.8567 & 37.8568 & 1e-06\tabularnewline
\hline 
18 & 33.5649 & 33.5650 & 2e-06\tabularnewline
\hline 
27 & 29.2731 & 29.2732 & 2e-06\tabularnewline
\hline 
36 & 24.9815 & 24.9815 & 3e-06\tabularnewline
\hline 
45 & 20.6928 & 20.6929 & 3e-06\tabularnewline
\hline 
54 & 16.4263 & 16.4264 & 2e-06\tabularnewline
\hline 
63 & 12.2536 & 12.2535 & 5e-06\tabularnewline
\hline 
72 & 8.3438 & 8.3436 & 2e-05\tabularnewline
\hline 
81 & 4.9680 & 4.9677 & 6e-05\tabularnewline
\hline 
90 & 2.4170 & 2.4168 & 9e-05\tabularnewline
\hline 
99 & 0.8472 & 0.8472 & 9e-07\tabularnewline
\hline 
108 & 0.1546 & 0.1547 & 1e-04\tabularnewline
\hline 
117 & 0.0023 & 0.0023 & 8e-05\tabularnewline
\hline 
120 & 0.0000 & 0.0000 & 0\tabularnewline
\hline 
\end{tabular}
\par\end{center}

\captionof{table}{\label{tab:Errors-on-a-strike-ladder} Results and relative errors for different strikes; barrier = 120, forward = 105.13.}%
\end{minipage}%
\end{minipage}

\begin{flushleft}
The results match with the desired accuracy.
%are conclusive and the Forward PIDE proves its usability and effectiveness to form part of a calibration routine.
\par\end{flushleft}

\section{Conclusions}
\label{sec:conclusions}

In this paper, we provided a new and numerically effective way to work with observed barrier option prices
to determine the Brunick-Shreve Markovian projection of a stochastic variance of a semi-martingale $S$ with running maximum $M$ onto $\left(S_{t},M_{t}\right)$.
%Our goal was to extend the idea of local volatility through concepts that are standard in the financial industry.
%We emphasise
%the potential of this approach to constitute the backbone of an 
%efficient calibration routine involving barrier quotes.
%Furthermore, we proved existence
%as well as uniqueness of the mimicking Brunick-Shreve coefficient
%in the case of up-and-out barrier options, and rigorously derived the forward equation satisfied by their prices.
The present Dupire-type formula, which provides similar advantages (and disadvantages) as the standard Dupire approach for vanillas,
was then re-arranged to obtain a forward PIDE, which is convenient to control the numerical stability of
best fit algorithms \cite{Crepey2010}. There is a well-known literature
on the calibration of vanilla options through Markovian projection
for LSV models (see Guyon and Labordère \cite{Guyon2011} and Ren, Madan and Quian \cite{Ren&Madan2007}).
We believe that extending these methods combined with the Forward
PIDE we presented in this paper will lead to novel calibration
algorithms for barrier options.
%This represents our main research interest for which we already have some encouraging conclusions.

%\newpage{}

%% file: mainresult.tex
\section{Setup and Main Result\label{sec:Setup-and-Main}}

We consider a market with a risk-free, deterministic and possibly time-dependent short rate, $r(t)$ at time $t$, and continuously compounded deterministic dividend $q(t)$.
Then 
$D\left(t\right)=\exp(-{\int_{0}^{t}r\left(u\right)\, du})$ is the discount factor, and 
$Q\left(t\right)=\exp({\int_{0}^{t}q\left(u\right)\, du})$ is the dividend capitalisation.

We assume the existence of a filtered probability space ({\Large $\chi$}, $\mathcal{F},\left\{ \mathcal{F}_{t}\right\} _{t\geq0},\mathbb{Q})$
with a (not necessarily unique) risk-neutral measure $\mathbb{Q}$, 
under which the price process of a risky asset follows
\begin{equation}
\frac{dS_{t}}{S_{t}}=\left(r\left(t\right)-q\left(t\right)\right)\, dt+\alpha_{t}\, dW_{t},
\label{eq:Model Definition}
\end{equation}
where $W$ is a standard Brownian motion and $\alpha_{t}$ is a continuous and positive $\mathcal{F}_{t}$-adapted semi-martingale such that
\begin{equation}
\mathbb{E}^{\mathbb{Q}}\left[\int_0^t \alpha_u^2 S_u^2\ du\right] < \infty.
\label{eq:alpha_condition}
\end{equation}

A practically important example is $\alpha_{t} = f\left(S_{t},t\right)\sqrt{V_{t}}$ with a \emph{local volatility} function $f$ and a CIR process $V$, which is often referred to as a \emph{local stochastic volatility} (LSV) model.

In this paper, we consider an up-and-out call option on $S$ with continuously monitored barrier $B$, strike $K$, and maturity $T$.
The arbitrage-free price under $\mathbb{Q}$ is
\begin{eqnarray*}
C\left(K,B,T\right)=\mathbb{E^{\mathbb{Q}}}\left[D\left(T\right)\left(S_{T}-K\right)^{+}\mathbf{1}_{M_{T}<B}\right],
\end{eqnarray*}
where $\mathbf{1}_{A}$ is the indicator function of event $A$ and
\[
M_{t}=\underset{0\leq u\leq t}{\max}S_{u}
\]
the running maximum process of $S$.
Adaptations of our results to other types of barrier options (such as put payoffs, down-and-out barriers etc) are easily possible using a similar derivation to below.
%, while any generalisation to discretely monitored barriers would look substantially different.

We now derive the first main result which links the barrier option price to
the Markovian projection of the stochastic volatility (\ref{eq:Model Definition}) onto the spot and its running maximum $\left(S_{t},M_{t}\right){}_{t\geq0}$,
which we define by
\begin{eqnarray}
\label{markpro}
\sigma_{|S,M}^{2}\left(K,B,T\right)=\mathbb{E^{\mathbb{Q}}}\left[\alpha_{T}^{2}\mid S_{T}=K, M_{T}=B\right].
\end{eqnarray}

A recent result by Brunick and Shreve \cite{BrunickShreve2013}, which we apply to the process $\log\left(S_{t}\right)$, shows that under the 
integrability condition (\ref{eq:alpha_condition}) on 
$\alpha$, the function $\sigma_{|S,M}$ is measurable, and gives the existence of a one-dimensional ``mimicking'' process $\widehat{S}$ 
with running maximum $\widehat{M}$,
such that for all $t$
\begin{eqnarray}
(S_t, M_t) \stackrel{law}{=} (\widehat{S}_t,\widehat{M}_t),
\label{eq:brunickSMinLaw}
\end{eqnarray}
where $\widehat{S}$ is the weak solution of
\begin{eqnarray}
\label{eqn:mimick}
\frac{d\widehat{S}_{t}}{\widehat{S}_{t}}=\left(r\left(t\right)-q\left(t\right)\right)\, dt+\sigma\left(\widehat{S}_{t},\widehat{M}_{t},t\right)\, d\widehat{W}_{t},
\label{eq:brunickMimickingModel}\end{eqnarray}
with a standard Brownian motion $\widehat{W}$ on a suitable probability space,
and
\begin{eqnarray*}
\sigma(K,B,T) = \sigma_{|S,M}(K,B,T).
\end{eqnarray*}

\begin{rem*}
We can take two views of (\ref{eqn:mimick}) from the perspective of derivative pricing. First, we can think of $\sigma$ in (\ref{eqn:mimick}) as a model in its own right, with a local volatility which additionally depends on the running maximum -- see \cite{Guyon2014} for a discussion of ``path-dependent'' volatility models. This extra flexibility of matching the path-dependence of the local volatility allows calibration to barrier contracts as well as vanillas, yet keeps the market complete.

Second, and more common, a different volatility model may be used, for instance the aforementioned LSV model. In that case, 
%if $\sigma^{2}\left(K,B,T\right)$equals $\sigma_{|S,M}^{2}\left(K,B,T\right)$ for all $K$, $B$ and $T$, then 
the mimicking Brunick-Shreve model will give the same barrier option prices as the
higher-dimensional diffusion as barrier option values are characterised precisely by the joint distribution of spot and running maximum.
This property is important for calibration purposes as we will describe later.
In either case, $\sigma$ plays a similar role for barrier options as the Dupire local volatility does for vanillas.
%In this spirit, our main result below, Theorem \ref{thm:MainResult}, extends the Dupire formula to the barrier case.
\end{rem*}

\begin{assn}
\label{assumption1}
%We assume that $\sigma$ in (\ref{eqn:mimick}) is such that the process $(\widehat{S},\widehat{M})$ is Markovian, e.g., that $\sigma$ is bounded and locally Lipschitz continuous(see \cite{Brunick2011}). Moreover, $\sigma$ is assumed to belong to the space $C^{2,1,0}(\Omega\times \mathbb{R}^{+})$. 
We assume that $\sigma$ in (\ref{eqn:mimick}) is bounded and $C^{2,1,0}(\Omega\times \mathbb{R}^{+})$. This implies that $(\widehat{S},\widehat{M})$ is Markovian (see \cite{Brunick2011}; bounded and locally Lipschitz is sufficient).
We also assume that $(S_t,M_t)$ -- or, equivalently,  $(\widehat{S}_t,\widehat{M}_t)$ -- has a
twice differentiable transition density $\phi$ under $\mathbb{Q}$, with bounded derivatives.
\end{assn}

 %with enough regularity 
 \begin{prop}
 \label{prop:kfe}
Under Assumption \ref{assumption1}, the density $\phi$ is the classical solution of the following Kolmogorov forward equation
in the region $\Omega = \{(x,y): 0<x<y, S_0<y\}$, 
\begin{eqnarray}
\label{eqn:kfe}
\frac{\partial \phi}{\partial t} + (r(t)-q(t)) \frac{\partial}{\partial x} (x \phi) - \frac{1}{2} \frac{\partial^2}{\partial x^2}\left( \sigma^2 x^2 \phi \right) &=& 0, \qquad \qquad \qquad\;
(x,y) \in \Omega, \; t>0, \\
\label{eqn:kfebc}
\frac{\partial}{\partial x}\left( \sigma^2 x^2 \phi \right)  + \frac{1}{2} \frac{\partial}{\partial y}\left( \sigma^2 x^2 \phi \right) &=& (r(t)-q(t))\  x\  \phi,
\qquad x=y, \; t>0, \!\! \\
%\nonumber
\label{zeroSzero}
\phi(S_0,S_0,t) &=& 0, \qquad \qquad \qquad \qquad \qquad \;\;\;\;\;\; t>0,\\
\nonumber
\phi(x,y,0) &=& \delta(x-S_0) \delta(y-S_0), \quad \quad (x,y) \in \Omega.
\end{eqnarray}
\end{prop}
For completeness, we give a derivation of these equations %under the assumption of smoothness %under the assumption of smoothness 
in Appendix \ref{sec:Derivation-of-the-Kolmogorov-Equations}.

This initial-boundary value problem is the adjoint to the backward equation satisfied by the option value $C$ as a function of $S_t$, $M_t$ and $t$, which has a homogeneous Neumann boundary condition on the diagonal $S=M$.

\begin{rem*}
The assumption of smoothness of the joint density is non-trivial. For the Black-Scholes model, with constant $\sigma$, the joint density function is known explicitly and smooth. In a recent work, \cite{Forde2013a} show the existence of the joint density -- not necessarily differentiable -- of an It{\^o} process $X$ and its running minimum $\underline{X}$, where the local volatility is a sufficiently smooth and suitably bounded function of $X_t$ and $\underline{X}_t$.
We also note that although the density features in our proof, it does not appear in the final result itself, and we conjecture that the regularity assumption may be weakened. Nonetheless, forward equations involving the transition density are widely used for calibration in practice and we envisage that (\ref{eqn:kfe}) can also be a useful building block for calibration in the present framework. Therefore, it seems reasonable to require that the model has enough regularity for the density to satisfy such an equation.
\end{rem*}

The time zero value of barrier options can then be written as
\begin{eqnarray}
\nonumber
C(K,B,T) &=& 
D(T) \; \int_{S_{0}}^\infty \int_0^\infty (x-K)^+ \mathbf{1}_{y<B} \, \phi(x,y,T) \dx \dy \\
&=& D(T) \; \int_{\KuS}^B \int_K^y (x-K) \, \phi(x,y,T) \dx \dy.
\label{eqn:intcall}
\end{eqnarray}

We can now extend the argument from Dupire \cite{Dupire} for European calls to derive a forward equation for barriers.
\begin{thm}
\label{thm:MainResult}
Under (\ref{eq:alpha_condition}) and Assumption \ref{assumption1}, %on the process given by (\ref{eq:Model Definition}),
the value $C(K,B,T)$ of an up-and-out barrier call satisfies, for all $0<K<B$ and $T>0$,
{\small
\begin{eqnarray}
\nonumber
\frac{\partial^{2}C\left(K,B,T\right)}{\partial B\partial T}+\left(r\left(T\right)-q\left(T\right)\right)K\frac{\partial^{2}C\left(K,B,T\right)}{\partial K\partial B} & = & \frac{1}{2}\sigma_{|S,M}^{2}\left(K,B,T\right)K^{2}\frac{\partial^{3}C\left(K,B,T\right)}{\partial K^{2}\partial B}-q\left(T\right)\frac{\partial C\left(K,B,T\right)}{\partial B} \vspace{0.1 cm} \\
 &  & -\frac{1}{2}\frac{\D}{\D B} \left(\left(B-K\right) B^{2}\sigma_{|S,M}^{2}\left(B,B,T\right)\left.\frac{\partial^{3}C\left(K,B,T\right)}{\partial K^{2}\partial B}\right\rfloor _{K=B}\right)\,,
%\end{array}
\label{eq:MainResult}
%\end{align}
\end{eqnarray}
}
where $\sigma_{|S,M}$ is given by (\ref{markpro}).
\end{thm}

\begin{proof}
Because of Brunick and Shreve's mimicking result, (\ref{eq:brunickSMinLaw})
and (\ref{eq:brunickMimickingModel}), and the smoothness assumption,
we can work with a density satisfying the forward equation (\ref{eqn:kfe})
and (\ref{eqn:kfebc}). We denote $\mu\left(t\right) = r\left(t\right)-q\left(t\right)$. Differentiating (\ref{eqn:intcall}) with
respect to $B$ and $T$ in the first equation, using (\ref{eqn:kfe})
in the second, and integrating by parts in the third, we obtain{\small{}
	\begin{eqnarray*}
		\frac{\partial^{2}C\left(K,B,T\right)}{\partial T\partial B} & = & -r\left(T\right)\frac{\partial C\left(K,B,T\right)}{\partial B}+D\left(T\right)\int_{K}^{B}(x-K)\frac{\partial\phi}{\partial t}(x,B,T)\dx\\
		& = & -r\left(T\right)\frac{\partial C\left(K,B,T\right)}{\partial B}+ \\ 
		&& D\left(T\right)\int_{K}^{B}(x-K)\left[-\mu\left(T\right)\frac{\partial}{\partial x}\left(x\phi(x,B,T)\right)+\frac{1}{2}\frac{\partial^{2}}{\partial x^{2}}\left(\sigma^{2}(x,B,T)x^{2}\phi(x,B,T)\right)\right]\dx\\
		& = & -r\left(T\right)\frac{\partial C\left(K,B,T\right)}{\partial B}+\\
		&  & \mu\left(T\right)D\left(T\right)\left[-B\left(B-K\right)\phi\left(B,B,T\right)+\int_{K}^{B}\left(x-K\right)\phi(x,B,T)\dx+K\int_{K}^{B}\phi(x,B,T)\dx\right]+\\
		&  & \frac{1}{2}D\left(T\right)\left(\sigma^{2}(K,B,T)K^{2}\phi(K,B,T)-\sigma^{2}(B,B,T)B^{2}\phi(B,B,T)\right)\ + \\
		&& \frac{1}{2}D\left(T\right)(B-K)\left[\frac{\partial}{\partial x}\left(\sigma^{2}(x,B,T)x^{2}\phi(x,B,T)\right)\right]_{x=B}.
	\end{eqnarray*}
} Using the boundary condition (\ref{eqn:kfebc}) in the second line
of {\small{}
	\begin{eqnarray*}
		\frac{{\rm d}}{{\rm d}B}\left(\sigma^{2}(B,B,T)B^{2}\phi(B,B,T)\right) & = & \left[\frac{\partial}{\partial x}\left(\sigma^{2}(x,B,T)x^{2}\phi(x,B,T)\right)\right]_{x=B}+\left[\frac{\partial}{\partial B}\left(\sigma^{2}(x,B,T)x^{2}\phi(x,B,T)\right)\right]_{x=B}\\
		& = & -\left[\frac{\partial}{\partial x}\left(\sigma^{2}(x,B,T)x^{2}\phi(x,B,T)\right)\right]_{x=B}+2\mu\left(T\right)B\phi\left(B,B,T\right),
	\end{eqnarray*}
} we get, by another application of the product rule, {\small{}
\begin{eqnarray*}
	\frac{{\rm d}}{{\rm d}B}\left((B-K)\sigma^{2}(B,B,T)B^{2}\phi(B,B,T)\right)&=&-(B-K)\left[\frac{\partial}{\partial x}\left(\sigma^{2}(x,B,T)x^{2}\phi(x,B,T)\right)\right]_{x=B}+ \\ && 2\mu\left(T\right)(B-K)B\phi\left(B,B,T\right) %\!\!\\
	+\sigma^{2}(B,B,T)B^{2}\phi(B,B,T).
\end{eqnarray*}
} The result now follows by differentiating (\ref{eqn:intcall}) once
more to obtain {\small{}
	\begin{eqnarray}
	\nonumber
		\frac{\partial^{2}C\left(K,B,T\right)}{\partial K\partial B}&=&-D\left(T\right)\int_{K}^{B}\phi\left(x,B,T\right)\dx, \\
			\label{eq:Barrier-Prices-Links-With-Density}
		\frac{\partial^{3}C\left(K,B,T\right)}{\partial K^{2}\partial B}&=&D\left(T\right)\phi\left(K,B,T\right),		
	\end{eqnarray}
} and substituting everything above.
\end{proof}

\begin{rem*}
We can think of (\ref{eq:MainResult}) in two ways. Given a model, either via $\sigma$ directly or via a specification that allows computation of the Markovian projection, the forward equation allows computation of barrier option prices across all strikes, maturities and barriers as the solution of a single PDE.
Conversely, if barrier call prices are observed on the market, (\ref{eq:MainResult}) allows inferences on the Markovian projection (\ref{markpro}) of the volatility. As is the case with the Dupire formula for European calls, a continuum of prices is not available and some sort of interpolation is required and can be notoriously unstable. We will return to these points in depth in the following sections.
\end{rem*}

For future reference, we define
\[
\widetilde{C}\left(K,B,T\right)=\mathbb{E^{\mathbb{Q}}}\left[D\left(T\right)Q\left(T\right)\left(S_{T}-K\right)^{+}\mathbf{1}_{M_{T}<B}\right],
\]
the capitalisation of the market price with dividends. 
By doing so, we eliminate the term $q\left(T\right)\frac{\partial C\left(K,B,T\right)}{\partial B}$ if we replace
$C\left(K,B,T\right)$ by $\widetilde{C}\left(K,B,T\right)$ in (\ref{eq:MainResult})$ $.
We will work with $\widetilde{C}$ in the remainder of the article.

%% file: derivation.tex
\section{Derivation of the Kolmogorov Forward and Backward Equations}
\label{sec:Derivation-of-the-Kolmogorov-Equations}

In this section, we derive the forward and backward Kolmogorov equations for a model of the type
%$\begin{cases}
\begin{eqnarray*}
\cfrac{dS_{t}}{S_{t}} &=& (r\left(t\right)-q\left(t\right)) \, dt+\sigma(S_{t},M_{t},t) \, dW_{t}, \\
M_{t} &=& \underset{0\leq u\leq t}{\max}S_{u}
\end{eqnarray*}
%\end{cases}$
and under Assumption \ref{assumption1}.

\subsection*{The Backward Equation (Derivation of Proposition \ref{prop:bw})}

We first note that the up-and-out call is a traded contract. As a consequence, the discounted price process $D(t){C}_{t}$ is a martingale under the risk-neutral measure $\mathbb{Q}$. Hence 
\[
D(t)C_{t}=\mathbb{E^{\mathbb{Q}}}\left[D(T)\left({S}_{T}-K\right)^{+}\mathbf{1}_{{M}_{T}<B} \mid \mathcal{F}_{t} \right] \qquad
 \; 0\le t\le T.
\]
Furthermore, as $\left(S_{t},M_{t}\right)$
is a Markovian vector, there exists a function $v$
such that
\[
D(t)v(x,y,t)=\mathbb{E^{\mathbb{Q}}}\left[D(T)\left({S}_{T}-K\right)^{+}\mathbf{1}_{{M}_{T}<B} \mid S_{t}=x,\, M_{t}=y \right] \qquad
\forall(x,y) \in \Omega_{B}, \; 0\le t\le T.
\]
We assume that $v$ is smooth and belongs to the space $C^{2,1,1}(\Omega_{B}\times \mathbb{R}^{+})$. 
Now write
\[
C_t = v(S_t,M_t,t).
\]
%We will show that the pricing backward equation for up-and-out call option is as follows
%\begin{eqnarray}
%\label{eq:Brunick_Backward_PDE2}
%\frac{\partial C}{\partial t}+\left(r\left(t\right)-q\left(t\right)\right)x\frac{\partial C}{\partial x}+\frac{1}{2}x^{2}\sigma^{2}\left(x,y,t\right)\frac{\partial^{2}C}{\partial x^{2}}-r\left(t\right)C=0,\, & (x,y) \in \Omega_{B}, \;0<t<T, \\
%C\left(0,y,t\right)=0, & t\leq T,\nonumber \\
%C\left(x,y,T\right)=\left(x-K\right)^{+}\mathbf{1}_{y<B}, & y>S_0, \nonumber \\
%\left.\frac{\partial C\left(x,y,t\right)}{\partial y}\right\rfloor _{x=y}=0, & y>S_0, \nonumber \\
%C\left(x,B,t\right)=0, & x>0,\, t\leq T. \nonumber
%\end{eqnarray}

By the Itô-Doeblin lemma and recalling that the running maximum process $M$ has finite variation, zero quadratic variation
and zero cross-variation with $S$,
\begin{eqnarray*}
d(D(t)C_{t}) & = & -r\left(t\right)D(t)C_{t} \, dt+D(t)\left[\frac{\partial v}{\partial t} \, dt+\frac{\partial v}{\partial x} \, dS_{t}+
\frac{\partial v}{\partial y} \, dM_{t} + \frac{1}{2}\frac{\partial^{2}v}{\partial x^{2}} \, d[S]_t \right].
\end{eqnarray*}
The process $D\left(t\right){C}_{t}$ is a martingale if and only if
\begin{eqnarray}
\label{eqn:bwpde}
\frac{\partial v}{\partial t}-r\left(t\right) v+(r\left(t\right)-q\left(t\right)) x \frac{\partial v}{\partial x}+\frac{1}{2}\sigma^{2} x^{2} \frac{\partial^{2} v}{\partial x^{2}} & = & 0, \quad (x,y) \in \Omega_{B}, \;0<t<T, \\
\left.\frac{\partial v}{\partial y}\right\rfloor _{x=y} & = & 0, \quad y>S_0.
\nonumber
\end{eqnarray}
This gives the desired backward PDE.
We note that the boundary condition is equivalent to
\begin{eqnarray}
\label{eqn:equivbc}
\frac{\D v(y,y,t)}{\dy}=\left.\frac{\partial v(x,y,t)}{\partial x}\right\rfloor _{x=y}.
\end{eqnarray}

\subsection*{The Forward Equation (Derivation of Proposition \ref{prop:kfe})}
We assume as before that the joint density function $\phi$ of the process $(S_{t},M_{t})$ exists and belongs to the space $C^{2,1,1}(\Omega\times \mathbb{R}^{+})$. 
%We will show that $\phi$ satisfies 
%\begin{eqnarray}
%	\frac{\partial\phi}{\partial t}+(r(t)-q(t))\frac{\partial}{\partial x}(x\phi)-\frac{1}{2}\frac{\partial^{2}}{\partial x^{2}}\left(\sigma^{2}x^{2}\phi\right) & = & 0,\qquad\qquad\qquad\;(x,y)\in\Omega,\; t>0,\label{eqn:kfe2}\\
%	\frac{\partial}{\partial x}\left(\sigma^{2}(x,y,t)x^{2}\phi\right)+\frac{1}{2}\frac{\partial}{\partial y}\left(\sigma^{2}(x,y,t)x^{2}\phi\right) & = & (r(t)-q(t))\, x\,\phi,\qquad x=y,\; t>0,\label{eqn:kfebc2}\\
%	\phi(S_{0},S_{0},t) & = & 0,\qquad\qquad\qquad\qquad\qquad\;\;\;\;\; t>0,\nonumber \\
%	\phi(x,y,0) & = & \delta(x-S_{0})\delta(y-S_{0}),\quad\quad(x,y)\in\Omega.\nonumber 
%\end{eqnarray}

Let $h \in C^{2,1}\left(\Omega\right)$ be a \emph{test} function of two (spatial) variables. We also assume that
$h$ and its derivatives vanish for $x=0$. We denote $\mu\left(t\right) = r\left(t\right)-q\left(t\right)$.

We apply the Itô-Doeblin lemma and get: 
\begin{eqnarray*}
	h\left(S_{t},M_{t}\right) & = & h\left(S_{0},M_{0}\right)\\
	&  & +\int_{0}^{t}\left[\mu\left(u\right)S_{u}\frac{\partial h\left(S_{u},M_{u}\right)}{\partial x}+\frac{1}{2}\sigma^{2}\left(S_{u},M_{u},u\right)S_{u}^{2}\frac{\partial^{2}h\left(S_{u},M_{u}\right)}{\partial x^{2}}\right]\du\\
	&  & +\int_{0}^{t}\sigma\left(S_{u},M_{u},u\right) S_{u} \frac{\partial h\left(S_{u},M_{u}\right)}{\partial x}\, dW_{u}+\int_{0}^{t}\frac{\partial h\left(S_{u},M_{u}\right)}{\partial y}\, dM_{u}\,.
\end{eqnarray*}
By taking expectations we can write 
\begin{eqnarray*}
	\mathbb{E^{\mathbb{Q}}}\left[h\left(S_{t},M_{t}\right)\right] & = & h\left(S_{0},M_{0}\right)+\mathbb{E}^{\mathbb{Q}}\left[\int_{0}^{t}\mu\left(t\right)S_{u}\frac{\partial h\left(S_{u},M_{u}\right)}{\partial x}\du\right]\\
	&  & +\mathbb{E}^{\mathbb{Q}}\left[\int_{0}^{t}\frac{1}{2}\sigma^{2}\left(S_{u},M_{u},u\right)S_{u}^{2}\frac{\partial^{2}h\left(S_{u},M_{u}\right)}{\partial x^{2}}\du\right]+\mathbb{E}^{\mathbb{Q}}\left[\int_{0}^{t}\frac{\partial h\left(S_{u},M_{u}\right)}{\partial y}\, dM_{u}\right]\,.
\end{eqnarray*}
%\[
%\]
Hence, we can introduce the density function $\phi$ of $(S,M)$, assumed twice differentiable,
\[
\mathbb{E^{\mathbb{Q}}}\left[h\left(S_{t},M_{t}\right)\right] = \int_{S_0}^\infty \int_0^y
 h(x,y) \phi(x,y,t)\dx\dy \,,
\]
differentiate with respect to $t$, and write for all $t>\text{0}$
\begin{eqnarray}
	\int_{S_{0}}^{\infty}\int_{0}^{y} h\frac{\partial\phi}{\partial t} \dx \dy=\int_{S_{0}}^{\infty}I_{\mu}(y)\dy+\frac{1}{2}\int_{S_{0}}^{\infty}I_{\sigma}(y)\dy+\frac{\D}{\D t}\mathbb{E}^{\mathbb{Q}}\left[\int_{0}^{t}\frac{\partial h\left(S_{u},M_{u}\right)}{\partial y}\, dM_{u}\right],\nonumber \\
	\label{eqn:intdefs}
\end{eqnarray}
where, explicitly with all arguments, 
\begin{eqnarray*}
	I_{\sigma}(y) & = & \int_{0}^{y}\left(\sigma^{2}\left(x,y,t\right)x^{2}\frac{\partial^{2}h\left(x,y\right)}{\partial x{}^{2}}\right)\phi(x,y,t)\dx\,,\\
	I_{\mu}(y) & = & \int_{0}^{y}\left(\mu\left(t\right)x\frac{\partial h(x,y)}{\partial x}\right)\phi(x,y,t)\dx\,.
\end{eqnarray*}
For $I_{\sigma}$, we can perform integration by parts twice to get
\begin{eqnarray*}
	I_{\sigma}(y) & = & \left.\frac{\partial h(x,y)}{\partial x}\right\rfloor _{x=y}\sigma^{2}(y,y,t)y^{2}\phi(y,y,t)-h(y,y)\left.\frac{\partial\sigma^{2}(x,y,t)x^{2}\phi(x,y,t)}{\partial x}\right\rfloor _{x=y}\\
	&  & +\int_{0}^{y} h(x,y)\frac{\partial^{2}\sigma^{2}(x,y,t)x^{2}\phi(x,y,t)}{\partial x^{2}}\dx\,.
\end{eqnarray*}
For the integral of $I_\sigma$, since
\[
\left.\frac{\partial h\left(x,y\right)}{\partial x}\right\rfloor _{x=y}=\frac{\D h\left(y,y\right)}{\D y}-\left.\frac{\partial h\left(x,y\right)}{\partial y}\right\rfloor _{x=y},
\]
we can integrate the first term by parts with respect to $y$, 
\begin{eqnarray*}
	\int_{S_{0}}^{\infty}I_{\sigma}(y)\dy & = & -\int_{S_{0}}^{\infty} h(y,y)\left[\left.\frac{\partial\sigma^{2}(x,y,t)x^{2}\phi(x,y,t)}{\partial x}\right\rfloor _{x=y}+\frac{\D\sigma^{2}(y,y,t)y^{2}\phi(y,y,t)}{\D y}\right]\dy\\
	&  & +\;\int_{S_{0}}^{\infty}\int_{0}^{y}h(x,y)\left[\frac{\partial^{2}\sigma^{2}(x,y,t)x^{2}\phi(x,y,t)}{\partial x^{2}}\right]\dx\dy \\
	 & & -\int_{S_{0}}^{\infty}\left.\frac{\partial h(x,y)}{\partial y}\right\rfloor _{x=y}\sigma^{2}(y,y,t)y^{2}\phi(y,y,t)\dy \\
	&  & -\; h\left(S_{0},S_{0}\right)\sigma^{2}\left(S_{0},S_{0},t\right)S_{0}^{2}\phi\left(S_{0},S_{0},t\right).
\end{eqnarray*}
For $I_{\mu}$, we integrate by parts once, 
\begin{eqnarray*}
	I_{\mu}(y) & = & \mu\left(t\right)h(y,y)y\phi(y,y,t)-\int_{0}^{y}h(x,y)\left[(r\left(t\right)-q\left(t\right))\frac{\partial x\phi(x,y,t)}{\partial x}\right]\dx.
\end{eqnarray*}
We insert in (\ref{eqn:intdefs}), {\small{}{} 
	\begin{eqnarray}
		&  & \int_{S_{0}}^{\infty}\int_{0}^{y} h(x,y)\left[\frac{\partial\phi(x,y,t)}{\partial t}+\mu\left(t\right)\frac{\partial x\phi(x,y,t)}{\partial x}-\frac{1}{2}\left(\frac{\partial^{2}\sigma^{2}(x,y,t)x^{2}\phi(x,y,t)}{\partial x^{2}}\right)\right]\dx\dy=\label{eqn:forwardintversion}\\
		&  & \int_{S_{0}}^{\infty} h(y,y)\left[\mu\left(t\right)y\phi(y,y,t)-\frac{1}{2}\left(\left.\frac{\partial\sigma^{2}(x,y,t)x^{2}\phi(x,y,t)}{\partial y}\right\rfloor _{x=y}+2\left.\frac{\partial\sigma^{2}(x,y,t)x^{2}\phi(x,y,t)}{\partial x}\right\rfloor _{x=y}\right)\right]\dy\nonumber \\
		&  & -\;\frac{1}{2}h\left(S_{0},S_{0}\right)\sigma^{2}\left(S_{0},S_{0},t\right)S_{0}^{2}\phi\left(S_{0},S_{0},t\right)\nonumber \\
		&  & -\int_{S_{0}}^{\infty}\left.\frac{\partial h(x,y)}{\partial y}\right\rfloor _{x=y}\sigma^{2}(y,y,t)y^{2}\phi(y,y,t)\dy +\frac{\D}{\D t}\mathbb{E}^{\mathbb{Q}}\left[\int_{0}^{t}\frac{\partial h\left(S_{u},M_{u}\right)}{\partial y}\, dM_{u}\right]. \nonumber
	\end{eqnarray}
}{\small \par}

Let us first consider all functions $h$ of compact support on $\Omega$.
Hence $h$ and its derivatives are assumed to additionally vanish
for $x=y$. Since $M_{t}$ only grows when $M_{t}=S_{t}$, the term
$\int_{0}^{t}\frac{\partial h\left(S_{u},M_{u}\right)}{\partial y}\, dM_{u}$
is zero. Then (\ref{eqn:forwardintversion}) becomes 
\[
\int_{S_{0}}^{\infty}\int_{0}^{y} h(x,y)\left[\frac{\partial\phi(x,y,t)}{\partial t}+\mu\left(t\right)\frac{\partial x\phi(x,y,t)}{\partial x}-\frac{1}{2}\left(\frac{\partial^{2}\sigma^{2}(x,y,t)x^{2}\phi(x,y,t)}{\partial x^{2}}\right)\right]\dx\dy=0\,,
\]
where we conclude that for all $\left(x,y,t\right)$ in $\left(\Omega\times\mathbb{R}_{*}^{+}\right)$
\[
\frac{\partial\phi(x,y,t)}{\partial t}+\mu\left(t\right)\frac{\partial x\phi(x,y,t)}{\partial x}-\frac{1}{2}\left(\frac{\partial^{2}\sigma^{2}(x,y,t)x^{2}\phi(x,y,t)}{\partial x^{2}}\right)=0\,.
\]
Let us now consider all functions $h$ which do not vanish at $x=y$
but only depend on the space variable $x$ such that we define $g\left(x\right)=h\left(x,y\right)$
with $g$ in the space $C^{2}$$\left(\mathbb{R}^{+}\right)$. In that case,
the terms $\int_{0}^{t}\frac{\partial h}{\partial y}\, dM_{u}$ and
$\int_{S_{0}}^{\infty}\left.\frac{\partial h(x,y)}{\partial y}\right\rfloor _{x=y}\left.\sigma^{2}y^{2}\phi\right\rfloor _{x=y} \dy$
vanish. From (\ref{eqn:forwardintversion}) we can write
\begin{eqnarray*}
	0=\int_{S_{0}}^{\infty} g(y)\left[\mu\left(t\right)y\phi(y,y,t)-\frac{1}{2}\left(\left.\frac{\partial\sigma^{2}(x,y,t)x^{2}\phi(x,y,t)}{\partial y}\right\rfloor _{x=y}+2\left.\frac{\partial\sigma^{2}(x,y,t)x^{2}\phi(x,y,t)}{\partial x}\right\rfloor _{x=y}\right)\right]\dy-\\
	\frac{1}{2}g\left(S_{0}\right)\sigma^{2}\left(S_{0},S_{0},t\right)S_{0}^{2}\phi\left(S_{0},S_{0},t\right).&
\end{eqnarray*}
If additionally we impose $g$ to vanish at $x=S_{0}$, then we conclude
that for all $\left(x,y,t\right)$ in $\left(\Omega\times\mathbb{R}_{*}^{+}\right)$

\[
\mu\left(t\right)y\phi(y,y,t)=\frac{1}{2}\left.\frac{\partial\sigma^{2}(x,y,t)x^{2}\phi(x,y,t)}{\partial y}\right\rfloor _{x=y}+\left.\frac{\partial\sigma^{2}(x,y,t)x^{2}\phi(x,y,t)}{\partial x}\right\rfloor _{x=y}\,.
\]
Finally, for all functions $g$ which do not vanish at $x=S_{0}$
we are left with 
\[
g\left(S_{0},t\right)\sigma^{2}\left(S_{0},S_{0},t\right)S_{0}^{2}\phi\left(S_{0},S_{0},t\right)=0\,,
\]
which concludes the proof.

\begin{rem*}
It is straightforward to generalise the approach to work directly under (\ref{eq:Model Definition}), say under a local-stochastic volatility model, to derive
forward and backward equations for the joint density of $S_t$, $M_t$ and the stochastic variance $V_t$. We find again the dual operator with respect to the $L_2$ inner product, now with three spatial dimensions.
\end{rem*}